\documentclass[11pt,a4paper]{article}

\usepackage{amsmath,amssymb,amsthm}
\usepackage{mathtools}
\usepackage{booktabs}
\usepackage{array}
\usepackage{graphicx}
\usepackage{tikz}
\usepackage[margin=1in]{geometry}
\usepackage{enumitem}
\usepackage{hyperref}

\newtheorem{theorem}{Theorem}

\newtheorem{lemma}[theorem]{Lemma}
\newtheorem{corollary}[theorem]{Corollary}

\theoremstyle{definition}
\newtheorem{definition}[theorem]{Definition}
\newtheorem{fact}[theorem]{Computed fact}
\theoremstyle{remark}
\newtheorem{remark}[theorem]{Remark}

\newcommand{\C}{\mathbb{C}}
\newcommand{\Z}{\mathbb{Z}}
\newcommand{\Q}{\mathbb{Q}}
\newcommand{\braket}[2]{\langle#1|#2\rangle}
\newcommand{\CK}{\mathrm{CK}\text{-}31}

\newcommand{\ncertified}{eleven}
\newcommand{\heegnercompletedmin}{$\le43$, not certified}
\newcommand{\goldenmin}{$\le51$, not certified}
\newcommand{\goldenbases}{28}
\newcommand{\supercompleted}{157 & 378 & 126}
\newcommand{\supermin}{$\le47$, not certified & 29}
\newcommand{\superfourcompleted}{601 & 1620 & 540}
\newcommand{\superfourmin}{$\le47$, not certified & }
\newcommand{\cubiccompleted}{361 & 840 & 280}
\newcommand{\cubicmin}{$\le60$, not certified & 27}

\title{The Algebraic Landscape of Kochen--Specker Sets\\in Dimension Three}
\author{Michael Kernaghan}
\date{Version 11 (23 November 2026)}

\begin{document}
\maketitle

\begin{abstract}
The classic three-dimensional Kochen--Specker (KS) proofs take their ray coordinates from a small
set of numbers, $\{0,\pm1,\pm\sqrt2\}$ for the 33 rays of Peres and $\{0,\pm1,\pm2\}$ for the 31
rays of Conway and Kochen.  We treat that set, the \emph{coordinate alphabet}, as a parameter:
an alphabet generates a finite pool of rays, and the pool is or is not KS-uncolorable.  We
determine the uncolorable alphabets in three continuous families and one discrete one.
(A)~Two-symbol alphabets $\{0,\pm1,\pm x\}$, $x$ complex: the pool is uncolorable only when $x$
or $1/x$ is one of $\pm2,\pm\sqrt2,\pm\sqrt{-2}$ or $x$ is a primitive third or sixth root of
unity, and no KS set with coordinates in such an alphabet has fewer than the 31 rays of Conway
and Kochen.  (B)~Real three-symbol alphabets $\{0,\pm1,\pm x,\pm y\}$, normalized to $1<x<y$:
the pool is uncolorable exactly when two letters have ratio $2$ or $\sqrt2$, or at eight
further points; again no KS set has fewer than 31 rays.  (C)~Conjugate-closed alphabets $\{0,\pm1,\pm x,\pm x^{*}\}$: the
pool is uncolorable exactly when $|x|^2=2$, or $|x|^2=\tfrac12$, or $x^{*}/x$ is a primitive
third or sixth root of unity, or $x\in\{\pm2,\pm\tfrac12\}$.  The uncolorable alphabets are
therefore not isolated.  On the circle $|x|^2=2$ a fixed list of 45 words is uncolorable at
every point and no shorter fixed list works at every point; it collapses to Peres's 33 rays at $x=\sqrt2$, and
at $x=(1+\sqrt{-7})/2$ the pool contains a 43-ray KS set.  The circle $|x|^2=\tfrac12$ carries
the 33-ray family of Gould and Aravind.  (D)~The pool of $0$ and the $n$th roots of unity is
uncolorable if and only if $6\mid n$.  Minimum KS subsets are certified by exact computation
for \ncertified\ named pools.  The Lean proof assistant checks (D) for $n\ge3$, the
uncolorability of the 45 words at every point of the circle, and the minima 31, 33, 33, 33, 43
and 47 of six pools.  An appendix lists the claims of earlier versions that are withdrawn.
\end{abstract}

\section{Introduction}

The Kochen--Specker (KS) theorem~\cite{KochenSpecker1967} shows that quantum observables cannot
be assigned definite values independently of the measurement context (for a review see
Ref.~\cite{Budroni2022}).  It is proved by exhibiting a finite set of rays in Hilbert space that
admits no $\{0,1\}$-valued assignment compatible with orthogonality.  Dimension three is the
smallest in which such sets exist.  (Smaller sets exist in higher dimensions, for example 20
vectors in dimension four~\cite{Kernaghan1994} and 36 in dimension
eight~\cite{KernaghanPeres1995}.)  In dimension three the smallest known set is the 31-ray
configuration of Conway and Kochen ($\CK$), reported by Peres~\cite{Peres1993}; the best lower
bound is 24, obtained for real vectors by Kirchweger, Peitl and
Szeider~\cite{KirchwegerPeitlSzeider2023} and for real and complex vectors by Li, Bright and
Ganesh~\cite{LiBrightGanesh2024}, improving the earlier 22 of Uijlen and
Westerbaan~\cite{UijlenWesterbaan2016}.  Trandafir and Cabello~\cite{TrandafirCabello2025} proved
that the orthogonality relations of $\CK$ have a unique realization up to unitary equivalence,
and that no set of 30 or fewer rays arises from one specified construction starting from the
Yu--Oh set; that 31 is the minimum without that assumption remains a conjecture.

\paragraph{The alphabet.}
Peres drew his two three-dimensional examples as cubes~\cite{Peres1991,Peres1993}: each ray joins
the centre of the cube to a marked point on its surface.  The two drawings differ in one respect.
In both, the marked points lie on a $5\times5$ grid on each face; for the 33-ray set the grid
lines sit at $0,\pm1,\pm\sqrt2$, and for $\CK$ at $0,\pm1,\pm2$ (Fig.~\ref{fig:triptych}).  The
list of coordinate values is thus a parameter of the construction, and one can ask what happens
when it is varied.  We call the list an \emph{alphabet} (Pavicic et
al.~\cite{Pavicic2019} speak of ``vector components''), a coordinate triple a \emph{word}, and the
set of all rays with coordinates in the alphabet its \emph{pool}.

\begin{figure}[t]
\centering
\includegraphics[width=\textwidth]{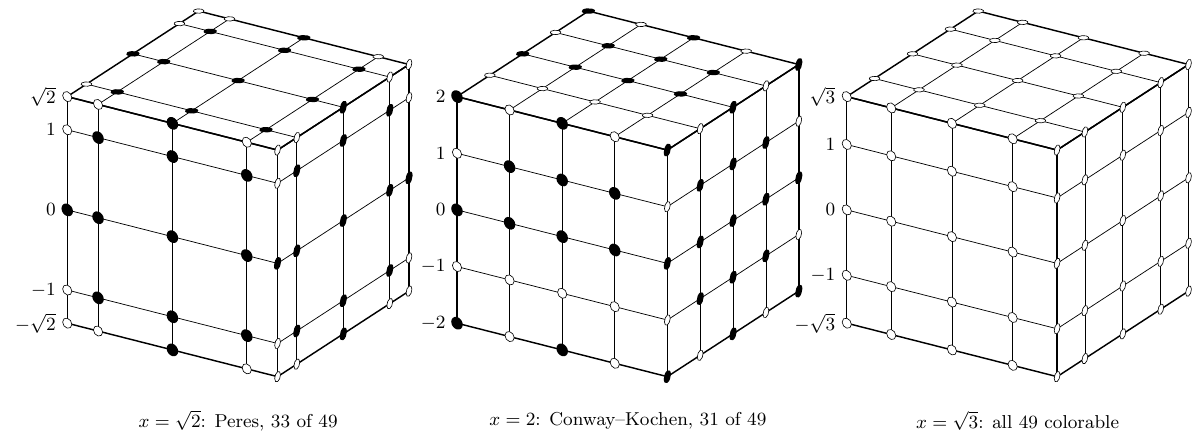}
\caption{One cube, three alphabets $\{0,\pm1,\pm x\}$ with $x$ real and $x>1$.  Each of the 49
rays of the pool leaves the centre through a grid point; antipodal points are the same ray.
Filled points: the 33 rays of Peres ($x=\sqrt2$, the unique smallest uncolorable subset of that
pool) and the 31 rays of Conway and Kochen ($x=2$, one of six smallest subsets, all equivalent
under signed permutations of the axes).  At $x=\sqrt3$ the whole pool is colorable.  Only the two
inner grid lines move with~$x$.  Nearness of two points on the cube says nothing about
orthogonality of the rays, and the picture does not represent the complex pools of
Section~\ref{sec:circle}.}
\label{fig:triptych}
\end{figure}

\begin{figure}[t]
\centering
\includegraphics[width=0.6\textwidth]{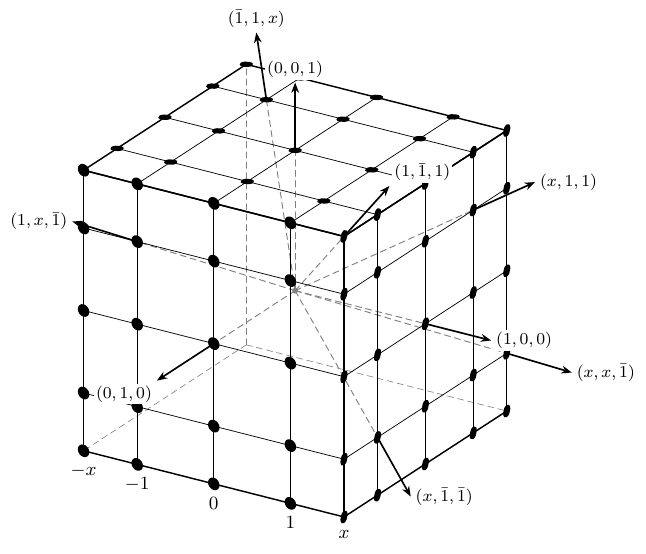}
\caption{Words as rays.  A ray whose coordinates form a word over $\{0,\pm1,\pm x\}$ leaves the
centre of the cube $[-x,x]^3$ through a grid point; a bar is a minus sign.  Words without the
letter~$x$, such as $(1,0,0)$ and $(1,\bar1,1)$, reach the surface after scaling by~$x$.  The
49 rays are the same list of words for every~$x$; which pairs are orthogonal is not.  For
example $(x,1,1)\cdot(x,\bar1,\bar1)=x^2-1-1$, which vanishes only for $x^2=2$.}
\label{fig:generic}
\end{figure}

\paragraph{Results.}
For a family of alphabets depending on parameters, the question is which parameter values give
an uncolorable pool and how small a KS set the pool then contains.  For three families we
classify colorability completely and determine every minimum except at two classes of points of
the third family.

\begin{enumerate}[label=(\roman*)]
\item \emph{Two-symbol alphabets $\{0,\pm1,\pm x\}$, $x$ complex} (Theorem~\ref{thm:twosymbol}).
The pool is uncolorable only at isolated points: $x$ or $1/x$ one of $\pm2$, $\pm\sqrt2$,
$\pm\sqrt{-2}$, or $x$ a primitive third or sixth root of unity.  No KS set with coordinates in
a two-symbol alphabet has fewer than 31 rays.
\item \emph{Real three-symbol alphabets $\{0,\pm1,\pm x,\pm y\}$} (Theorem~\ref{thm:threesymbol}).
After normalization to $1<x<y$, the pool is uncolorable on six curves, where two letters have
ratio 2 or~$\sqrt2$, and at eight further points.  Again no KS set has fewer than 31 rays.
\item \emph{Conjugate-closed alphabets $\{0,\pm1,\pm x,\pm x^{*}\}$, $x$ complex}
(Theorem~\ref{thm:conj}, Fig.~\ref{fig:circle}).  The pool is uncolorable exactly on two
circles, four lines through the origin with the origin removed, and four further points.  On
the circle $|x|^2=2$ one fixed list of 45 words is uncolorable at every point and no shorter
fixed list is
(Theorem~\ref{thm:circle}); the circle $|x|^2=\tfrac12$ carries the 33-ray family of Gould and
Aravind.
\end{enumerate}

A fourth family has a discrete parameter: the pool of $0$ and the $n$th roots of unity is
uncolorable exactly when $6\mid n$ (Theorem~\ref{thm:6n}).  Section~\ref{sec:catalogue} adds
pools that lie outside these families, each with exact counts and either a certified minimum or
a stated upper bound.

\paragraph{Standard of proof.}
Each of the three continuous-parameter classifications reduces, by an argument given in the
text, to finitely many cases that are decided in exact arithmetic by named scripts.  The
roots-of-unity theorem has a structural proof for every~$n$, with one finite computation at
$n=6$.  In addition the following are proved in the Lean
proof assistant and checked by its kernel (Section~\ref{sec:methods}): Theorem~\ref{thm:6n} for
$n\ge3$; the uncolorability of the 45 words for every $x$ on the circle; and the minimum KS
subsets of six pools (31, 33, 33, 33, 43 and 47 rays).

\paragraph{Relation to earlier versions.}
Versions 1--9 of this paper proposed a classification by ``cancellation mechanisms'' and
described the uncolorable alphabets as isolated ``islands''.  That proposal did not survive
checking and was withdrawn in version~10; Appendix~\ref{sec:corrections} lists the withdrawn
claims.  The present version restates the results that replaced it.  Since the uncolorable
alphabets are not isolated, we do not use the word ``island''.

\section{Definitions and methods}\label{sec:prelim}

\begin{definition}
A \emph{ray} in $\C^3$ is a one-dimensional subspace, written $[v]$ for a nonzero vector~$v$.
Rays $[v],[w]$ are \emph{orthogonal} if $\braket vw=\sum_k v_k^{*}w_k=0$.  A \emph{triad} (or
\emph{basis}) is a set of three mutually orthogonal rays.
\end{definition}

\begin{definition}
A \emph{coloring} of a set $S$ of rays is a map $f:S\to\{0,1\}$ such that (a)~no two orthogonal
rays of $S$ both receive~1, and (b)~every triad contained in $S$ has exactly one ray with value~1.
$S$ is \emph{KS-uncolorable}, or a \emph{KS set}, if it has no coloring.  Its \emph{orthogonality
graph} has the rays as vertices and the orthogonal pairs as edges; the triads are its triangles.
Both the edges and the triangles enter the definition, and an edge lying in no triangle of $S$
is still a constraint.
\end{definition}

A coloring of $S$ restricts to a coloring of any subset, because every orthogonal pair and every
triad of the subset is one of~$S$:

\begin{lemma}[Restriction]\label{lem:restrict}
A set of rays that contains an uncolorable subset is uncolorable.
\end{lemma}

\begin{definition}
An \emph{alphabet} is a finite set $\mathcal A\subset\C$ containing~0.  Its \emph{raw pool} is
\[
S(\mathcal A)=\{\,[v]: v\in\mathcal A^3\setminus\{0\}\,\}.
\]
The \emph{completed pool} $\bar S(\mathcal A)$ is the smallest set of rays containing
$S(\mathcal A)$ and containing, with every orthogonal pair $[v],[w]$, the third ray of their
triad, $[(v\times w)^{*}]$ (the ordinary cross product for real vectors).  In all cases reported
here the completion was computed to closure and is finite; we do not know a general finiteness
criterion.
\end{definition}

In a completed pool every orthogonal pair lies in exactly one triad, so the numbers of pairs $E$
and triads $T$ satisfy $E=3T$.  This is a convenient check: the raw Eisenstein pool has
$E=174$, $T=22$, so it is not closed, contrary to a statement in version~9.

\paragraph{Three cautions on scope.}
(1)~Uncolorability is a property of a pool, not of a ring or a field.  Every ring that appears in
this paper contains $\Z$ and therefore contains the coordinates of $\CK$; the statement ``this
field supports a KS set'' is true of all of them and distinguishes nothing.  (2)~The size of a
generator is not an invariant of the pool: multiplying every word by $x^{-1}$ shows
$S(\{0,\pm1,\pm x\})=S(\{0,\pm1,\pm x^{-1}\})$, and $S$ is also unchanged by $x\mapsto-x$.
(3)~A raw pool and its completion are different objects.  Their counts and minima need not
agree; inclusion gives only one-way implications (if the raw pool is uncolorable so is its
completion, and a KS subset of the raw pool bounds the minimum of the completion from above).
Every count below says which pool is meant.

\subsection{Methods}\label{sec:methods}

\paragraph{Exact arithmetic.}  Coordinates are elements of a number field $\Q[t]/(m(t))$ stored
as rational coefficient vectors; rays are compared after division by the first nonzero
coordinate; a pair is orthogonal when the Hermitian product is the zero element.  The
cyclotomic pools use exponents modulo~$n$, and the proof of Theorem~\ref{thm:circle} uses
integer polynomial identities.  No count in Sections~\ref{sec:real}--\ref{sec:catalogue} depends
on a floating-point tolerance.  The only floating-point quantities in the paper are the flex
counts of Section~\ref{sec:rigidity}; they are labelled as numerical.  Floating-point
evaluation is used elsewhere only as a cross-check of exact computations.

\paragraph{Colorability.}  Colorability of a finite set is decided by a SAT
solver~\cite{AudemardSimon2018}: one variable per ray, a clause forbidding two 1s on each
orthogonal pair, and a clause requiring at least one 1 on each triad.

\paragraph{Certified minima.}  Let $P$ be an uncolorable pool.  If $M\subseteq P$ is colorable
then, by Lemma~\ref{lem:restrict}, every uncolorable subset of $P$ contains a ray outside~$M$.
The certifier keeps a list of such sets $M$, each verified colorable by the SAT solver, and asks
a second solver for a set $R$ of fewer than $u$ rays meeting the complement of every listed~$M$,
where $u$ is the size of the best uncolorable subset found so far.  If $R$ is uncolorable, $u$
decreases; if $R$ is colorable it is enlarged to a colorable set $M\supseteq R$ and $M$ joins the
list.  When no such $R$ exists, no uncolorable subset has fewer than $u$ rays, and the minimum is
exactly~$u$.  The answer is correct whenever the SAT solver is.  The list of sets $M$ is in principle a
certificate that can be checked independently.  For six pools it is checked in Lean (below);
for Theorem~\ref{thm:circle} it is archived; the other minima are reproducible by re-running
the scripts.  Minima that the procedure did not finish within
its time budget are reported as upper bounds.

Two reductions speed up the runs for Theorems~\ref{thm:conj} and~\ref{thm:threesymbol}; both
keep the answer exact.  First, if a ray of an uncolorable set lies in no triad of that set,
deleting it leaves an uncolorable set (a coloring of the rest extends by giving the ray~0).  So a
minimum set lies in the \emph{core} of the pool, obtained by deleting rays that lie in no triad
until none is left, and the search is run there.  Second, if
$g$ is a permutation of the rays preserving orthogonality and $M$ is colorable, so is $g(M)$, so
each colorable $M$ may be added to the list together with all its images under a group of such
permutations (signed permutations of the coordinates, and complex conjugation where the alphabet
is closed under it; each is checked to preserve the graph before use).

The same enumeration, continued after the first
minimum set is found, lists \emph{all} minimum sets of a pool.

\paragraph{Proofs checked in Lean.}  Three files in the repository are checked by the kernel of
Lean~4 (with the mathlib library) and use no axioms beyond the standard three.  Where a proof
rests on a finite set of rays, their coordinates are exact elements of a ring written as
integer tuples, orthogonality is recomputed inside Lean, and the certificates, which are
computed outside Lean, are checked there and not trusted.
\texttt{Cyclotomic6n.lean} proves Theorem~\ref{thm:6n} for $n\ge3$.
\texttt{NamedSets.lean} proves that the 43-ray set, the 47-ray set and, for every $x$ with
$|x|^2=2$, the 45 words of Theorem~\ref{thm:circle} are uncolorable, each from a case-split
refutation.  \texttt{Minimality.lean} proves minima.  Its certificate for a pool with minimum
$u$ is the list of colorable sets $M$ described above, each with a coloring, together with a
branching proof that no set of fewer than $u$ rays meets the complement of every~$M$: a node
puts one ray in or out, and a leaf names complements that avoid the rays put in and are pairwise
disjoint once the rays put out are removed, so that a set meeting all of them needs a separate
ray for each.  A soundness theorem for this checker is proved once.  The result for each of six
pools, those of $\{0,\pm1,\pm2\}$, $\{0,\pm1,\pm\sqrt2\}$, $\{0,\pm1,\pm\sqrt{-2}\}$, the
Eisenstein alphabet, $\{0,\pm1,\pm\alpha,\pm\alpha^{*}\}$ with $\alpha=(1+\sqrt{-7})/2$, and
$\{0,\pm1,\pm\psi,\pm\psi^2\}$ with $\psi^3=\psi^2+1$, is a theorem with two halves: some $u$ of
the pool's vectors form an uncolorable set, and every set of fewer than $u$ of them has a
coloring.  Three limits should be stated.  First, in Lean each pool is an explicit list of
vectors; that the list consists of all rays with coordinates in the alphabet is not proved
there.  Second, the Lean statements are about sets of vectors, not rays.  This loses nothing,
by an argument that is not itself formalized: a coloring of rays gives a coloring of any choice
of representative vectors and conversely, and each listed pool has one vector per ray (the
export script checks this), so the minimum for vectors is the minimum for rays.  Third, these are minima within a pool; they
say nothing about KS sets outside it.

The procedure used in earlier versions (described there as ``OCUS'') was different and was not
sound: when a candidate was colorable, it required the next candidate to contain a ray that
completes a violated triad with \emph{two} rays of the current candidate, ignoring triads with
two or three rays outside it.  That can exclude genuine uncolorable subsets, and it did: it
reported 52 for the completed golden-ratio pool, which contains an uncolorable subset of~51.
None of the minima in this version comes from that code.

\section{Two-symbol alphabets}\label{sec:real}

The alphabets of Peres and of Conway and Kochen are both of the form $\{0,\pm1,\pm x\}$.  For any
$x\notin\{0,\pm1\}$ this alphabet has $5^3-1=124$ words, the same list for every~$x$
(Fig.~\ref{fig:generic}).  They give 49 rays for every $x$ other than $\pm i$, where more
words coincide and 31 remain (Theorem~\ref{thm:twosymbol}).  What depends on $x$ is which
inner products vanish.

\subsection{Which inner products can vanish}

A term of the inner product of two words is $a^{*}b$ with $a,b$ letters, so it is zero or one of
$\pm1,\ \pm x,\ \pm x^{*},\ \pm|x|^2$.  An inner product is therefore
\[
c_0+c_1x+c_2x^{*}+c_3|x|^2,\qquad c_i\in\Z,\quad |c_0|+|c_1|+|c_2|+|c_3|\le 3 .
\]
Some of these vanish for every $x$ (all $c_i=0$: disjoint supports, or terms cancelling in
pairs).  The others vanish only for special~$x$, and since there are finitely many coefficient
vectors they can be listed completely.  There are 128 nonzero vectors.  The pool is unchanged by
$x\mapsto-x$ and by $x\mapsto1/x$, and complex conjugation of all coordinates carries it to the
pool of $x^{*}$, which has the same orthogonality graph.  Together with an overall sign these
maps act on the coefficient vectors, inversion as
$(c_0,c_1,c_2,c_3)\mapsto(c_3,c_2,c_1,c_0)$.  The 128 vectors fall into 21 orbits, of which 13
have no solution other than $x\in\{0,\pm1\}$.  Table~\ref{tab:zerosums} lists the remaining
eight.

\begin{table}[t]
\centering
\caption{Every relation of at most three terms that can make an inner product of two words over
$\{0,\pm1,\pm x\}$ vanish for some but not all $x$, up to the symmetries $x\mapsto-x$,
$x\mapsto x^{*}$, $x\mapsto1/x$ and overall sign (computed by \texttt{ks\_zero\_sums\_v10.py}).
The row $x-x^{*}=0$ is the reason real alphabets have orthogonalities that generic complex ones
lack.
The rows $2+x=0$ and $1+2x=0$ describe the same pool, since $S$ is unchanged by $x\mapsto 1/x$.
A row says where extra orthogonalities appear; it says nothing by itself about colorability.}
\label{tab:zerosums}
\begin{tabular}{@{}rlll@{}}
\toprule
Terms & Representative relation & Orbit size & Solution set \\
\midrule
2 & $1-|x|^2=0$ & 2 & $|x|^2=1$ \\
2 & $x+x^{*}=0$ & 2 & $\mathrm{Re}\,x=0$ \\
2 & $x-x^{*}=0$ & 2 & $\mathrm{Im}\,x=0$ \\
3 & $2+x=0$ & 16 & $x=-2$ \\
3 & $2-|x|^2=0$ & 4 & $|x|^2=2$ \\
3 & $1+2x=0$ & 16 & $x=-1/2$ \\
3 & $1+x+x^{*}=0$ & 8 & $\mathrm{Re}\,x=-1/2$ \\
3 & $1+x-|x|^2=0$ & 8 & $x=(1\pm\sqrt{5})/2$ \\
\bottomrule
\end{tabular}

\end{table}

Two remarks.  First, the solution sets are not isolated points in general: two of them are
circles and three are lines in the complex plane.  Second, the relation $x+x^{*}=0$ was pointed
out to us by Frembs~\cite{Frembs2026}.  It holds on the imaginary axis, as $x-x^{*}=0$ holds on
the real axis, and the two classical modulus-2 alphabets differ in exactly this.  Of the 120
orthogonal pairs of the $\sqrt{-2}$ pool (in the word representatives used here), 60 hold for
every $x$, 54 use $x+x^{*}=0$ and 6 use $|x|^2=2$; at $x=\sqrt2$ the 54 are replaced by pairs
using $x-x^{*}=0$.  The difference is visible in the graphs (Computed fact~\ref{fact:pools}).
A table of this kind says where extra orthogonalities appear.  Which loci give uncolorable
pools has to be computed.

\subsection{Real alphabets}

\begin{table}[t]
\centering
\caption{Raw pools of two-symbol alphabets $\{0,\pm1,\pm x\}$, exact arithmetic.  ``Min'' is the
certified size of the smallest uncolorable subset.  Every real $x=\sqrt d$ with $d$ a non-square
integer, $3\le d\le30$, gives the same counts as $\sqrt3$, as do $x=3$, $x=4$ and
$x=(1+\sqrt d)/2$ for $d=13,17,21,29$.}
\label{tab:alphabets}
\begin{tabular}{lrrrcc}
\toprule
Alphabet & Rays & Pairs & Triads & Colorable & Min \\
\midrule
$\{0,\pm1\}$ & 13 & 24 & 4 & yes & --- \\
$\{0,\pm1,\pm2\}$ (also $\pm\tfrac12$) & 49 & 138 & 26 & \textbf{no} & \textbf{31} \\
$\{0,\pm1,\pm\sqrt2\}$ & 49 & 120 & 16 & \textbf{no} & \textbf{33} \\
$\{0,\pm1,\pm\sqrt3\}$ & 49 & 114 & 10 & yes & --- \\
$\{0,\pm1,\pm\varphi\}$, $\varphi=(1+\sqrt5)/2$ & 49 & 138 & 10 & yes & --- \\
\midrule
$\{0,\pm1,\pm\sqrt{-1}\}$ & 31 & 69 & 7 & yes & --- \\
$\{0,\pm1,\pm\sqrt{-2}\}$ & 49 & 120 & 16 & \textbf{no} & \textbf{33} \\
$\{0,\pm1,\pm\sqrt{-d}\}$, $d=3,5,6,\dots,15$ & 49 & 114 & 10 & yes & --- \\
\bottomrule
\end{tabular}
\end{table}

Table~\ref{tab:alphabets} gives the raw pools.  The 13-ray pool of $\{0,\pm1\}$ is colorable, as
it must be: 13 is below the lower bound of~24.  Among real two-symbol alphabets the
only uncolorable raw pools are those of $x=2$ and $x=\sqrt2$ (and their negatives and
reciprocals), the two alphabets of Peres's cubes; Theorem~\ref{thm:twosymbol} below proves this
for every real~$x$ and settles the complex case as well.
In the row for $\sqrt{-1}$ every nonzero letter is a unit, the group $\{\pm1,\pm i\}$ acts freely
on the 124 words, and $124/4=31$.

\begin{fact}[Integer pool]\label{fact:integer}
The 49-ray pool of $\{0,\pm1,\pm2\}$ has minimum uncolorable subsets of 31 rays, certified.  There
are exactly six such subsets.  They form a single orbit under the 48 signed permutations of the
coordinates, so all six are copies of $\CK$.  Together they use 37 of the 49 rays; the 12 unused
rays are those of squared length~9, such as $(1,2,2)$; 13 rays, those of $\{0,\pm1\}$, lie in all
six.  The completed pool has 109 rays, 258 pairs and 86 triads, and its minimum is also 31,
certified.
\end{fact}

The pool also contains uncolorable subsets
that are critical (no ray can be removed) but not of minimum size, among them the 33-ray set
of Sch\"utte, also known as the Conway--Kochen 33-ray set and communicated by them to
Peres~\cite{Peres1991} (the identification is made in~\cite{TrandafirCabello2025}), which has 76
orthogonal pairs and
20 bases and is contained in the 37 rays above.

None of three enlargements we tried lowers the minimum: the raw pools of $\{0,\pm1,\pm2,\pm3\}$ (145 rays),
$\{0,\pm1,\pm\sqrt2,\pm2\}$ (109 rays) and $\{0,\pm1,\pm\varphi,\pm2\}$ (145 rays) all have
certified minimum 31, with 17 bases in the minimum set found, while
$\{0,\pm1,\pm\varphi,\pm\varphi^2\}$ (109 rays, 24 triads) is colorable.  Each of the first three
contains the integer pool, so 31 is an upper bound by Lemma~\ref{lem:restrict}; the content of
the computation is that nothing smaller appears.

\begin{fact}[Peres pool]\label{fact:peres}
The 49-ray pool of $\{0,\pm1,\pm\sqrt2\}$ has a unique minimum uncolorable subset, the 33 rays of
Peres (72 pairs, 16 bases).  The other 16 rays lie in no triad.  The completed pool has 145 rays,
336 pairs and 112 triads, and its minimum is 33, certified.  The same statements hold for
$\{0,\pm1,\pm\sqrt{-2}\}$.
\end{fact}

\begin{fact}[The two modulus-2 pools are not the same graph]\label{fact:pools}
The minimum 33-ray sets of the $\sqrt2$ and $\sqrt{-2}$ pools have isomorphic orthogonality
graphs.  The two 49-ray pools do not, although they agree in the number of rays, pairs, triads,
in degree sequence and in spectrum; neither do the two 145-ray completed pools.  (Graph
isomorphism was tested with VF2~\cite{Cordella2004VF2}; for the raw pools we also checked directly
that none of the 48 isomorphisms between the two 33-ray cores carries the neighbourhoods of the
other 16 rays of one pool onto those of the other.)
\end{fact}

The difference is the one Frembs identified: 54 of the 120 orthogonalities of the $\sqrt{-2}$
pool use $x+x^{*}=0$, a relation that fails at $x=\sqrt2$, and he builds an arrangement on that
relation~\cite{Frembs2026}.

\subsection{All two-symbol alphabets}\label{sec:twosymbol}

The completeness of Table~\ref{tab:zerosums} turns the question ``for which $x$ is the pool of
$\{0,\pm1,\pm x\}$ uncolorable?'' into a finite computation, and the answer is a theorem about
every complex~$x$.

\begin{theorem}[Two-symbol alphabets]\label{thm:twosymbol}
Let $x$ be a nonzero complex number and $S(x)$ the pool of all words over $\{0,\pm1,\pm x\}$.
\begin{enumerate}[label=(\roman*)]
\item $S(x)$ is uncolorable if and only if $x$ or $1/x$ is one of $\pm2$, $\pm\sqrt2$,
$\pm\sqrt{-2}$, or $x$ is a primitive third or sixth root of unity.
\item The minimum uncolorable subset of $S(x)$ has 31 rays for $x\in\{\pm2,\pm\tfrac12\}$, where
the minimum sets are the six copies of $\CK$ of Computed fact~\ref{fact:integer}, and 33 rays in
every other uncolorable case.
\item Consequently, if all rays of a KS set have coordinates in $\{0,\pm a,\pm b\}$ for some
complex numbers $a$ and $b$, the set has at least 31 rays.
\end{enumerate}
\end{theorem}

\begin{proof}
Write $x=a+ib$.  The orthogonality graph of $S(x)$ depends on $x$ only through two finite pieces
of data: which words are proportional, and which inner products vanish.

\emph{Proportionality.}  Two words are proportional when their three $2\times2$ minors vanish.
A minor is a difference of two products of two letters, and such a product is $0$, $\pm1$, $\pm x$ or $\pm x^2$.  So a minor is identically zero, or a
monomial or twice a monomial (never zero), or one of $\pm(1\pm x)$, $\pm(1\pm x^2)$,
$\pm x(1\pm x)$.  Away from $x\in\{\pm1,\pm i\}$
the same words are proportional for every~$x$, and the pool has 49 rays.

\emph{Orthogonality.}  An inner product is $c_0+c_1x+c_2x^{*}+c_3|x|^2$ for one of the 128
coefficient vectors counted above, or is identically zero.  It vanishes when
$(c_1-c_2)\,b=0$ and $c_0+(c_1+c_2)\,a+c_3(a^2+b^2)=0$.  For a given vector the solution set is
therefore empty, a finite set of real points, the real axis, a vertical line, or a circle
centred on the real axis.  Going through the 128 vectors gives nine curves,
\[
\mathrm{Im}\,x=0,\quad \mathrm{Re}\,x=0,\ \pm\tfrac12,\quad |x|^2=\tfrac12,\ 1,\ 2,\quad
|x\pm1|=1,
\]
and finitely many real points.  Call a point \emph{special} if it is one of those real points, a
point of intersection of two of the curves, or one of $\pm1,\pm i$.  There are 27 nonzero special points
up to complex conjugation.  At a point that is not special and lies on none of the curves, no
vector vanishes.  At a point of a curve $C$ that is not special, the vectors that vanish are
exactly those whose solution set is~$C$, since any other vector vanishes on a different curve or
on special points only.  Hence the graph of $S(x)$ is constant on the complement of the curves,
constant on each curve away from the special points, and otherwise has to be computed at each
special point: $1+9+27=37$ cases.  (Conjugating all coordinates carries $S(x)$ to $S(x^{*})$ and
preserves orthogonality, so $x$ and $x^{*}$ need not be distinguished.)

\emph{Computation.}  The script \texttt{ks\_two\_symbol\_theorem.py} lists the curves and special
points, takes an exact sample point on each curve, builds the 37 graphs in exact arithmetic, decides colorability of each, and runs the certification procedure of
Section~\ref{sec:methods} on every uncolorable one; all terminate.
Table~\ref{tab:twosymbol} is the result.  Special points are grouped by the symmetries $x\mapsto-x$
and $x\mapsto1/x$, which do not change the pool; curves are listed on one row when their counts
agree, which does not mean their graphs are isomorphic (the generic graphs of the real axis, the
imaginary axis and the unit circle are pairwise non-isomorphic).  Parts (i) and (ii) are read off the table.  For
(iii), if $a$ or $b$ is zero or $a=\pm b$ the rays lie in the 13-ray pool of $\{0,\pm1\}$, which
is colorable; otherwise dividing all coordinates by $a$ puts the set inside $S(b/a)$.
\end{proof}

\begin{table}[t]
\centering
\caption{Every pool of a two-symbol alphabet $\{0,\pm1,\pm x\}$, $x\in\mathbb{C}\setminus\{0\}$
(Theorem~\ref{thm:twosymbol}).  A curve stands for its points other than the special points
listed below it.  ``Min'' is the certified size of the smallest uncolorable subset, and the last
column names the graph of the minimum set found (isomorphism by VF2).}
\label{tab:twosymbol}
\small
\begin{tabular}{@{}lrrrcl@{}}
\toprule
$x$ & Rays & Pairs & Triads & Min & Minimum set \\
\midrule
on no curve & 49 & 60 & 4 & --- & colorable \\
$\mathrm{Im}\,x=0$;\ \ $\mathrm{Re}\,x=0$;\ \ $|x|=1$ & 49 & 114 & 10 & --- & colorable \\
$\mathrm{Re}\,x=\pm\tfrac12$;\ \ $|x\pm1|=1$ & 49 & 72 & 8 & --- & colorable \\
$|x|^2=2$;\ \ $|x|^2=\tfrac12$ & 49 & 66 & 10 & --- & colorable \\
\midrule
$\pm1$ & 13 & 24 & 4 & --- & colorable \\
$\pm i$ & 31 & 69 & 7 & --- & colorable \\
$\pm\varphi,\ \pm1/\varphi$ & 49 & 138 & 10 & --- & colorable \\
$\pm1\pm i,\ (\pm1\pm i)/2$ & 49 & 78 & 14 & --- & colorable \\
$(\pm1\pm\sqrt{-7})/2,\ (\pm1\pm\sqrt{-7})/4$ & 49 & 78 & 14 & --- & colorable \\
$\pm2,\ \pm\tfrac12$ & 49 & 138 & 26 & \textbf{31} & $\CK$ \\
$\pm\sqrt2,\ \pm1/\sqrt2$ & 49 & 120 & 16 & \textbf{33} & Peres \\
$\pm\sqrt{-2},\ \pm1/\sqrt{-2}$ & 49 & 120 & 16 & \textbf{33} & Peres \\
$(\pm1\pm\sqrt{-3})/2$ & 49 & 138 & 18 & \textbf{33} & Eisenstein 33 \\
\bottomrule
\end{tabular}
\end{table}

Four remarks on what the theorem does and does not say.

First, it is a statement about a class of coordinate systems, not about KS sets in general.
The minimum number of rays of a KS set in dimension three is known only to lie between 24
and~31, and a set smaller than 31, if one exists, would need at least three nonzero letters up to
sign.  What the theorem
shows is that the construction read off Peres's cubes, one unit and one further number along
each axis, cannot beat Conway and Kochen whatever the number is.

Second, it concerns the raw pool.  Completion adds rays whose coordinates are not letters, and
the completed pool of $\{0,\pm1,\pm\varphi\}$ is uncolorable although its raw pool is not
(Section~\ref{sec:catalogue}).

Third, the uncolorable cases are isolated points; no curve of Table~\ref{tab:zerosums} is
uncolorable for a two-symbol alphabet.  In particular the two-symbol pool is colorable at a
generic point of the circle $|x|^2=2$, and even at $x=1+i$ and $x=(1+\sqrt{-7})/2$, where it
meets another curve.  The uncolorability of that whole circle (Theorem~\ref{thm:circle}) needs the
conjugate letter $x^{*}$.

Fourth, at a primitive sixth root of unity three curves
meet ($|x|=1$, $\mathrm{Re}\,x=\tfrac12$, $|x-1|=1$), and the 49 rays over $\{0,\pm1,\pm\omega\}$
are already uncolorable, with a minimum set isomorphic to the Eisenstein 33-ray set of
Section~\ref{sec:roots}.  The letter $\omega^2$ of the Eisenstein alphabet is not needed.  We do
not claim this set is new: it has the graph of Cabello's 33-ray set~\cite{Cabello2025simplest}.

Here ``coordinates in an alphabet'' means, as throughout, that every ray has a representative
whose three entries are letters, in one common orthonormal basis.  We have not compared
Theorem~\ref{thm:twosymbol} with the discussion of alphabets in the preprints of Flores
Gordillo~\cite{FloresGordillo2026} and make no claim of priority for the classification.

\section{Real three-symbol alphabets}\label{sec:threesymbol}

The same method works with one more real letter.  Signs are in the alphabet and dividing every
word by a letter does not change the pool, so a real alphabet $\{0,\pm a,\pm b,\pm c\}$ whose
letters have three distinct nonzero absolute values may be written $\{0,\pm1,\pm x,\pm y\}$ with
$1<x<y$; if two absolute values agree or a letter is zero, the alphabet has at most two symbols and
Theorem~\ref{thm:twosymbol} applies.

\begin{theorem}[Real three-symbol alphabets]\label{thm:threesymbol}
Let $1<x<y$ be real and $S(x,y)$ the pool of all words over $\{0,\pm1,\pm x,\pm y\}$.
\begin{enumerate}[label=(\roman*)]
\item $S(x,y)$ is uncolorable if and only if two of the letters $1,x,y$ have ratio $2$ or
$\sqrt2$, or $(x,y)$ is one of the eight points of Table~\ref{tab:threesymbol}.
\item Its minimum uncolorable subset has 31 rays when two letters have ratio~2, 33 rays when two
letters have ratio $\sqrt2$ and none has ratio~2, and between 42 and 57 rays at the eight points
(Table~\ref{tab:threesymbol}).
\item Consequently, if all rays of a KS set have coordinates in $\{0,\pm a,\pm b,\pm c\}$ for some
real numbers $a,b,c$, the set has at least 31 rays.
\end{enumerate}
\end{theorem}

\begin{proof}
An inner product of two words, and a $2\times2$ minor of two words, is
$p_c=c_0+c_1x+c_2y+c_3x^2+c_4xy+c_5y^2$ for an integer vector~$c$ with
$\sum_k|c_k|\le3$; 376 nonzero vectors occur, counting $c$ and $-c$ separately.  Factor each $p_c$ over~$\Q$;
125 distinct irreducible factors appear.  Since $p_c(P)=0$ exactly when one of its factors vanishes
at~$P$, which words are proportional, and hence which words represent the same ray, and which
are orthogonal depend only on the set $F(P)$ of irreducible factors that vanish at~$P$; so
does the orthogonality graph.  The sets $F(P)$ that occur for $P$ in the domain $D=\{1<x<y\}$ are:
the empty set; $\{f\}$ for each factor $f$ whose zero set meets $D$ in a curve (31 of them); and
the set $F(P)$ at each point of $D$ that lies on two of these curves or is a singular point of
one.  There are 130 such special sets.  The script \texttt{ks\_three\_symbol\_real.py} makes
every one of these decisions exactly, holding real algebraic numbers as rational isolating
intervals and deciding signs by exact interval arithmetic.  It finds the curves by cutting the
line $x>1$ at the real roots of the leading coefficient and discriminant in~$y$ and of
$f(x,x)$: on each cell the real branches over~$x$ neither meet nor cross $y=x$, so one rational
point per cell decides whether a branch lies in~$D$.  It finds the special points by splitting
each finite intersection, and the singular locus of each curve, into prime components
$\{(h_x(t),h_y(t)) : q(t)=0\}$ with $q$ irreducible, using a separating coordinate
$t=x+sy$ after making the ideal radical.  A factor $h$ vanishes on a component exactly
when $q$ divides $h(h_x,h_y)$.  It then builds the $1+31+130=162$ graphs and decides
colorability.  As a cross-check it evaluates every special point to 60 digits.  It also
rebuilds the graph (the partition of the words into rays and the set of orthogonal pairs) at a
point of each of the 162 strata, this time from the coordinates and without the factorisation,
and finds it identical in every case.  An earlier version of the scan that made the same
decisions in floating point gave the same 162 cases.

Six curves are uncolorable: $x=2$, $y=2$, $y=2x$, $x^2=2$, $y^2=2$, $y^2=2x^2$.  The generic
point, the other 25 curves and 65 of the special sets are colorable; the other 65 special sets
are uncolorable, and all but eight of them lie on one of the six curves.  Every uncolorable case
was certified (Section~\ref{sec:methods}), each run also proved separately that no
uncolorable subset has fewer than 31 rays, and a minimum set for each case is archived and
checked uncolorable.  Points of the six curves have the minimum of the
curve, except that a point on a ratio-2 curve and a ratio-$\sqrt2$ curve has minimum~31.  This
gives (i) and~(ii).  For (iii), if two of $|a|,|b|,|c|$ agree or one is zero, the set lies in a
two-symbol pool and Theorem~\ref{thm:twosymbol}(iii) applies; otherwise, dividing by the smallest
absolute value and renaming puts it inside some $S(x,y)$.
\end{proof}

\begin{figure}[t]
\centering
\begin{tikzpicture}[x=3.4cm,y=1.75cm]
\fill[gray!12] (1,1) -- (3.05,3.05) -- (3.05,4.1) -- (1,4.1) -- cycle;
\draw[->,gray] (0.9,1) -- (3.2,1) node[right] {$x$};
\draw[->,gray] (1,0.9) -- (1,4.3) node[above] {$y$};
\foreach \v in {1,2,3} \draw[gray] (\v,1) -- (\v,0.94) node[below,font=\small] {$\v$};
\foreach \v in {1,2,3,4} \draw[gray] (1,\v) -- (0.97,\v) node[left,font=\small] {$\v$};
\draw[gray,dashed] (1,1) -- (3.05,3.05) node[right,font=\small] {$y=x$};
% ratio 2 (minimum 31)
\draw[line width=1.1pt] (2,2) -- (2,4.1);
\draw[line width=1.1pt] (1,2) -- (2,2);
\draw[line width=1.1pt] (1,2) -- (2.05,4.1);
% ratio sqrt 2 (minimum 33)
\draw[line width=0.5pt] (1.4142,1.4142) -- (1.4142,4.1);
\draw[line width=0.5pt] (1,1.4142) -- (1.4142,1.4142);
\draw[line width=0.5pt] (1,1.4142) -- (2.899,4.1);
\node[font=\scriptsize,anchor=south west] at (2.06,4.1) {$x=2$};
\node[font=\scriptsize,anchor=south] at (1.4142,4.1) {$x=\sqrt2$};
\node[font=\scriptsize,anchor=north west] at (1.72,3.4) {$y=2x$};
\node[font=\scriptsize,anchor=west] at (2.92,4.05) {$y=\sqrt2\,x$};
\node[font=\scriptsize,anchor=east] at (0.9,2.12) {$y=2$};
\node[font=\scriptsize,anchor=east] at (0.9,1.47) {$y=\sqrt2$};
% the eight further points, labelled by their minimum
\foreach \px/\py/\m/\a in {1.272/1.618/42/north west, 1.466/2.148/47/south west,
  1.883/2.132/49/south east, 2.197/2.414/49/north west, 2.206/2.659/51/south east,
  2.359/2.783/55/west, 2.732/3.732/56/north west, 1.250/1.600/57/south east}
  {\fill (\px,\py) circle (0.9pt); \node[font=\scriptsize,anchor=\a,inner sep=2pt] at (\px,\py) {\m};}
\end{tikzpicture}
\caption{Real three-symbol alphabets $\{0,\pm1,\pm x,\pm y\}$, $1<x<y$ (shaded).  The pool is
uncolorable exactly on the six lines and at the eight points (Theorem~\ref{thm:threesymbol}).
Heavy lines: two letters have ratio 2, minimum 31.  Light lines: ratio $\sqrt2$, minimum 33
(31 where a light line meets a heavy one).  Each point is labelled with its minimum
(Table~\ref{tab:threesymbol}); the points labelled 42 and 57 are distinct but close.}
\label{fig:threesymbol}
\end{figure}

\begin{table}[t]
\centering
\caption{The eight points of Theorem~\ref{thm:threesymbol} that lie on none of the six
uncolorable curves, with the relations that hold there (each point is the unique point of $D$
satisfying them) and the certified minimum.  Rays, pairs and triads are those of the raw pool.}
\label{tab:threesymbol}
\small
\begin{tabular}{@{}llrrrr@{}}
\toprule
Relations & $(x,y)\approx$ & Rays & Pairs & Triads & Min \\
\midrule
$x^2=y,\ y^2=y+1$\ \ ($x=\sqrt\varphi,\ y=\varphi$) & $(1.272,1.618)$ & 109 & 312 & 28 & 42 \\
$x^3=x^2+1,\ y=x^2$\ \ ($x=\psi,\ y=\psi^2$) & $(1.466,2.148)$ & 109 & 348 & 40 & 47 \\
$y^2=x^2+1,\ xy=x+y$ & $(1.883,2.132)$ & 145 & 474 & 42 & 49 \\
$y^2=x^2+1,\ x^2=2y$\ \ ($y=1+\sqrt2$) & $(2.197,2.414)$ & 145 & 402 & 46 & 49 \\
$xy=x^2+1,\ y=x^2-x,\ y^2=x^2+x$ & $(2.206,2.659)$ & 145 & 462 & 46 & 51 \\
$xy=x^2+1,\ x^2=2y$ & $(2.359,2.783)$ & 145 & 414 & 58 & 55 \\
$y=x+1,\ x^2=2y$\ \ ($x=1+\sqrt3$) & $(2.732,3.732)$ & 145 & 486 & 46 & 56 \\
$xy=2,\ y^2=x^2+1$ & $(1.250,1.600)$ & 145 & 390 & 46 & 57 \\
\bottomrule
\end{tabular}
\end{table}

Three remarks.  First, the curves contribute nothing new: on a ratio-2 curve the pool contains
the words over $\{0,\pm1,\pm2\}$ (after rescaling) and hence $\CK$, and on a ratio-$\sqrt2$ curve
it contains Peres's set; the theorem says that no third letter, placed anywhere on these curves,
gives a smaller set.  Second, the eight points are exceptional uncolorable alphabets (we have
not searched the literature for them), and all their minima exceed~33.  One of them is the supergolden alphabet of Section~\ref{sec:catalogue}, whose
minimum 47 was certified there independently.  Third, the theorem is about real letters; the
complex three-symbol case has the same finite structure in principle but a parameter space of four
real dimensions, and we have not attempted it.

\section{Conjugate-closed alphabets}\label{sec:conj}

The two uncolorable alphabets $\{0,\pm1,\pm\sqrt2\}$ and $\{0,\pm1,\pm\sqrt{-2}\}$ both have
$|x|^2=2$, and in both the conjugate $x^{*}=\pm x$ is already a letter.  For a general complex
$x$ the alphabet that is closed under conjugation is
\[
\mathcal A(x)=\{0,\pm1,\pm x,\pm x^{*}\},\qquad P(x)=S(\mathcal A(x)).
\]
For generic $x$ the pool $P(x)$ has 145 rays; at special $x$ some of them coincide (49 remain
at $x=\sqrt{\pm2}$).  This section determines every $x$ at which $P(x)$ is uncolorable
(Theorem~\ref{thm:conj}, Fig.~\ref{fig:circle}): two circles, four lines through the origin
with the origin removed, and four points.  It
then takes the two circles in turn.  On $|x|^2=2$ one fixed list of 45 words is uncolorable at
every point (Theorem~\ref{thm:circle}).  On $|x|^2=\tfrac12$ the pool contains the 33-ray family
of Gould and Aravind (Section~\ref{sec:halfcircle}).

\begin{figure}[t]
\centering
\begin{tikzpicture}[scale=1.55]
\draw[->,gray] (-2.45,0) -- (2.6,0) node[right] {$\mathrm{Re}\,x$};
\draw[->,gray] (0,-2.1) -- (0,2.25) node[above] {$\mathrm{Im}\,x$};
\foreach \t in {30,60,120,150} \draw[line width=0.9pt] (\t:-1.9) -- (\t:1.9);
\draw[line width=0.9pt] (0,0) circle (1.4142);
\draw[line width=0.9pt] (0,0) circle (0.7071);
\foreach \p in {-2,-0.5,0.5,2} \fill (\p,0) circle (0.05);
\draw[fill=white,line width=0.5pt] (0,0) circle (0.06);
\node[font=\small,anchor=south] at (2,0.08) {$2$: $\CK$};
\fill (1.4142,0) circle (0.05) node[below right,font=\small] {$\sqrt2$: Peres};
\fill (0,1.4142) circle (0.05);
\node[font=\small,anchor=south west] at (0.08,1.78) {$\sqrt{-2}$: 33 rays};
\draw[gray,line width=0.3pt] (0.03,1.46) -- (0.12,1.8);
\fill (0.5,1.3229) circle (0.05);
\node[font=\small,anchor=west] at (1.2,1.45) {$\tfrac{1+\sqrt{-7}}{2}$: minimum 43};
\draw[gray,line width=0.3pt] (0.55,1.34) -- (1.18,1.44);
\node[font=\small,anchor=east] at (-1.25,1.5) {$|x|^2=2$};
\node[font=\small,anchor=east] at (-1.75,-0.5) {$|x|^2=\tfrac12$};
\draw[gray,line width=0.3pt] (-0.69,-0.17) -- (-1.73,-0.48);
\node[font=\small,anchor=west] at (1.68,-0.97) {$\arg x=-30^\circ$};
\node[font=\small,anchor=north west] at (0.97,-1.66) {$\arg x=-60^\circ$};
\node[font=\small,align=left,anchor=north east] at (-1.2,-1.5) {elsewhere:\\colorable};
\end{tikzpicture}
\caption{The alphabet $\{0,\pm1,\pm x,\pm x^{*}\}$.  Its pool is uncolorable exactly on the set
drawn in black: the circles $|x|^2=2$ and $|x|^2=\tfrac12$, the four lines through the origin at
$\pm30^\circ$ and $\pm60^\circ$ to the real axis with the origin removed, and the points
$\pm2,\pm\tfrac12$ (Theorem~\ref{thm:conj}).  The pool is unchanged by $x\mapsto-x$ and
$x\mapsto x^{*}$, so one quadrant suffices.  At the four axis points of the outer circle, which
are two alphabets up to sign, the 45 words of Table~\ref{tab:words} collapse to 33 rays.  At $(1+\sqrt{-7})/2$ the additional
relation $x+x^{*}=1$ adds orthogonalities and the minimum drops from 45 to~43.}
\label{fig:circle}
\end{figure}

\subsection{The classification}\label{sec:conjclosed}

\begin{theorem}[Conjugate-closed alphabets]\label{thm:conj}
Let $x$ be a nonzero complex number and $P(x)$ the pool of all words over
$\{0,\pm1,\pm x,\pm x^{*}\}$.
\begin{enumerate}[label=(\roman*)]
\item $P(x)$ is uncolorable if and only if one of the following holds:
\begin{enumerate}[label=(\alph*)]
\item $|x|^2=2$;
\item $|x|^2=\tfrac12$;
\item $x^{*}/x$ is a primitive third or sixth root of unity, that is, $\arg x$ is one of
$\pm30^\circ$, $\pm60^\circ$, $\pm120^\circ$, $\pm150^\circ$;
\item $x\in\{\pm2,\pm\tfrac12\}$.
\end{enumerate}
\item At every point of the curves (a)--(c) other than the special points of
Table~\ref{tab:conj}, the minimum uncolorable subset of $P(x)$ has 45 rays on $|x|^2=2$ and 33
rays on $|x|^2=\tfrac12$ and on the four lines of~(c).  At the special points the minima are
those of Table~\ref{tab:conj}; two of them are bounds and not certified values.
\item For every $x$ other than $\pm1\pm i$, no uncolorable subset of $P(x)$ has fewer than 31
rays, and 31 is attained at $x=\pm2,\pm\tfrac12$, where $P(x)$ is the integer pool.  At
$x=\pm1\pm i$ the pool is uncolorable and we have not established the bound 31; there the
minimum lies between the general lower bound 24 and~45.
\end{enumerate}
\end{theorem}

\begin{proof}
The argument is that of Theorem~\ref{thm:twosymbol} with two more monomials.  An inner product of
two words, and a $2\times2$ minor of two words, is an integer combination
\[
c_0+c_1x+c_2x^{*}+c_3|x|^2+c_4x^2+c_5x^{*2}
\]
with $\sum|c_k|\le3$; 376 nonzero coefficient vectors occur.  With $x=a+ib$ the combination
vanishes when
\[
b\,\bigl[(c_1-c_2)+2(c_4-c_5)\,a\bigr]=0,\qquad
c_0+(c_1+c_2)\,a+(c_3+c_4+c_5)\,a^2+(c_3-c_4-c_5)\,b^2=0 .
\]
Both conditions involve $b$ only through the factor $b$ and through $b^2$, and the second is
linear in~$b^2$.  The solution set of one vector is therefore made up of finitely many points
and of whole curves of three kinds: the real axis, a vertical line, and a curve
$b^2=g_0+g_1a+g_2a^2$.  We work throughout in the coordinates $(a,b^2)$, which identifies $x$
with $x^{*}$; this loses nothing, because $P(x^{*})=P(x)$.  Two different curves of these kinds
meet in at most two values of~$a$.  A vector that vanishes at infinitely many points of a curve
$b^2=g(a)$ with $b\ne0$ vanishes on all of it: the first condition then forces
$(c_1-c_2)+2(c_4-c_5)a$ to be identically zero, and substituting $b^2=g(a)$ in the second
gives a polynomial in $a$ with infinitely many roots.  So the two lines of a pair such as
$b^2=3a^2$, and the two branches of a hyperbola, need not be treated separately.
Going through the 376 vectors gives fourteen curves with points other than the origin:
\[
\mathrm{Im}\,x=0,\quad \mathrm{Re}\,x=0,\ \pm\tfrac12,\quad |x|^2=\tfrac12,\ 1,\ 2,\quad
|x\pm1|=1,\quad \mathrm{Re}\,x^2=0,\ \pm\tfrac12,\quad x^2+x^{*2}=\pm|x|^2 .
\]
The last two are the pairs of lines $b^2=a^2/3$ and $b^2=3a^2$ of part~(c).  Call a nonzero point
\emph{special} if it is an isolated solution of some vector or lies on two of the curves; there
are 53 special points with $\mathrm{Im}\,x\ge0$ (92 in the plane, 14 of them real).  As in Theorem~\ref{thm:twosymbol}, the set of vectors that vanish at $x$,
and with it the list of coinciding words and the orthogonality graph of $P(x)$, is the same at
all points on no curve, and the same at all non-special points of any one curve.  That leaves
$1+14+53=68$ cases.  The pool is also unchanged by $x\mapsto-x$, which reduces
the special points to 28 classes and the curves to 12; unlike the two-symbol pool it is not
unchanged by $x\mapsto1/x$.

The script \texttt{ks\_conjugate\_closed\_theorem.py} lists the curves and the special points in
exact arithmetic ($a$ and $b^2$ in one quadratic field $\Q(\sqrt d)$), builds the 68 graphs, decides colorability of
each, and runs the certification procedure of Section~\ref{sec:methods} on the uncolorable ones.
Table~\ref{tab:conj} is the result.  The uncolorable cases are the curves (a), (b) and (c), every
special point on them, and the points $\pm2,\pm\tfrac12$; all other cases are colorable.  That
is~(i).  Part~(ii) is the corresponding rows of the table.  For~(iii), the certifier terminated
with a minimum of at least 31 in every uncolorable case except four.  For those four we made
separate runs, each of which only decides whether an uncolorable subset of fewer than a given
number of rays exists.  At $\eta/\sqrt2$ and at $2\zeta$
(notation of Table~\ref{tab:conj}) there is none of fewer than 33 rays,
which with the 33-ray sets found makes 33 the minimum.  At $\zeta/\sqrt2$ there is none of
fewer than 32, and the run for 33 did not terminate in eight hours.  At $x=\pm1\pm i$ the run
for 31 did not terminate in eight hours.  The
witnesses are archived; the lists of colorable sets produced by these runs are not, so the
lower bounds of this theorem can be reproduced only by rerunning the script.

Part~(c) also has a proof that needs no new computation.  A word whose nonzero entries are all
$\pm x$ or $\pm x^{*}$ is $x$ times a word over $\{0,\pm1,\pm u\}$ with $u=x^{*}/x$, so $P(x)$
contains every ray of the two-symbol pool $S(u)$.  By Theorem~\ref{thm:twosymbol}, $S(u)$ is
uncolorable when $u$ is a primitive third or sixth root of unity, with a 33-ray minimum set
having the graph of the Eisenstein set, and Lemma~\ref{lem:restrict} carries this to $P(x)$.
This argument gives nothing more, since $S(u)$ is colorable for every other $u$ on the unit
circle.
\end{proof}

\begin{table}[t]
\centering
\caption{Every pool of the conjugate-closed alphabet $\{0,\pm1,\pm x,\pm x^{*}\}$
(Theorem~\ref{thm:conj}).  A curve stands for its points other than the special points listed
below it, and each special point stands for its images under $x\mapsto-x$ and
$x\mapsto x^{*}$.  $\zeta=e^{i\pi/3}$, $\eta=e^{i\pi/6}$, $\varphi=(1+\sqrt5)/2$.  ``Min'' is
the certified size of the smallest uncolorable subset; ``32--33'' means that 33 rays suffice
and fewer than 32 do not; ``$\le45$'' is an upper bound only.  The last column gives the graph
of the smallest set found (isomorphism by VF2).}
\label{tab:conj}
\small
\begin{tabular}{@{}lrrrcl@{}}
\toprule
$x$ & Rays & Pairs & Triads & Min & Smallest set found \\
\midrule
on no curve & 145 & 366 & 22 & --- & colorable \\
$\mathrm{Im}\,x=0$;\ \ $\mathrm{Re}\,x=0$ & 49 & 114 & 10 & --- & colorable \\
$|x|=1$ & 109 & 276 & 16 & --- & colorable \\
$\mathrm{Re}\,x^2=0$ & 127 & 321 & 19 & --- & colorable \\
$\mathrm{Re}\,x=\pm\tfrac12$;\ \ $|x\pm1|=1$ & 145 & 462 & 30 & --- & colorable \\
$\mathrm{Re}\,x^2=\pm\tfrac12$ & 145 & 378 & 22 & --- & colorable \\
$|x|^2=2$ & 145 & 378 & 34 & \textbf{45} & 45 rays, 98 pairs \\
$|x|^2=\tfrac12$ & 145 & 390 & 46 & \textbf{33} & Peres \\
$\arg x=\pm30^\circ,\pm60^\circ,\pm120^\circ,\pm150^\circ$ & 145 & 390 & 30 & \textbf{33} & Eisenstein 33 \\
\midrule
$1$ & 13 & 24 & 4 & --- & colorable \\
$i$ & 31 & 69 & 7 & --- & colorable \\
$\varphi,\ 1/\varphi$ & 49 & 138 & 10 & --- & colorable \\
$e^{i\pi/4}$ & 91 & 231 & 13 & --- & colorable \\
$|x-1|=1$ and $\mathrm{Re}\,x^2=\tfrac12$ & 145 & 474 & 30 & --- & colorable \\
\midrule
$2,\ \tfrac12$ & 49 & 138 & 26 & \textbf{31} & $\CK$ \\
$\sqrt2,\ 1/\sqrt2,\ \sqrt{-2},\ 1/\sqrt{-2}$ & 49 & 120 & 16 & \textbf{33} & Peres \\
$(1+i)/2$ & 127 & 465 & 51 & \textbf{33} & Peres \\
$(1+\sqrt{-7})/4$ & 145 & 486 & 54 & \textbf{33} & Peres \\
$\eta/\sqrt2$ & 145 & 414 & 54 & \textbf{33} & Peres \\
$\zeta/\sqrt2$ & 145 & 414 & 54 & 32--33 & Peres \\
$\zeta$ & 57 & 174 & 22 & \textbf{33} & Eisenstein 33 \\
$\eta$ & 109 & 312 & 24 & \textbf{33} & Eisenstein 33 \\
$\zeta/2$ & 145 & 414 & 54 & \textbf{33} & Eisenstein 33 \\
$\eta/\sqrt3,\ \sqrt3\,\eta$ & 145 & 486 & 38 & \textbf{33} & Eisenstein 33 \\
$\sqrt2\,\eta,\ \sqrt2\,\zeta$ & 145 & 402 & 42 & \textbf{33} & Eisenstein 33 \\
$2\zeta$ & 145 & 438 & 62 & \textbf{33} & Eisenstein 33 \\
$(1+\sqrt{-7})/2$ & 145 & 522 & 42 & \textbf{43} & 43 rays, 107 pairs \\
$(\sqrt5+\sqrt{-3})/2,\ (\sqrt3+\sqrt{-5})/2$ & 145 & 390 & 34 & \textbf{45} & 45 rays, 98 pairs \\
$1+i$ & 127 & 477 & 63 & $\le45$ & 45 rays, 98 pairs \\
\bottomrule
\end{tabular}
\end{table}

Five remarks.

First, the theorem says where $P(x)$ is uncolorable and not why.  For the circle $|x|^2=2$,
Section~\ref{sec:circle} gives a reason: a fixed list of 45 words whose orthogonalities follow
from $xx^{*}=2$ alone (Theorem~\ref{thm:circle}).  Part~(ii) of that theorem says that $P(x)$ is
colorable outside a finite union of proper real algebraic subsets; the present theorem
identifies the union as two circles, four lines without the origin, and four points.

Second, the three curves behave differently.  On $|x|^2=2$ the minimum is 45 except at
finitely many points.  On $|x|^2=\tfrac12$ it is 33 except possibly at the one class of
points marked 32--33, and the smallest set found has the orthogonality graph of Peres's set at
a generic point and at every special point of that circle (Section~\ref{sec:halfcircle}).
Every point of the lines other than the origin has a
33-ray uncolorable set with the graph of the Eisenstein set of Section~\ref{sec:roots}, written
in the letters $\{0,\pm x,\pm x^{*}\}$, by the last paragraph of the proof; it is a minimum set
wherever the table certifies~33.

Third, earlier versions surveyed the alphabets $\{0,\pm1,\pm a,\pm a^{*}\}$ with $a$ the
generator of the ring of integers of $\Q(\sqrt{-d})$, for the nine imaginary quadratic fields
of class number one, and looked for an arithmetic reason why $d=2,3,7$ give uncolorable pools.
The theorem accounts for the survey without reference to number fields: $d=2$ and $d=7$ are
points of $|x|^2=2$; $d=3$ is the point $\zeta$ of the $60^\circ$ line; $d=1$ is $x=i$,
colorable; and for $d=11,19,43,67,163$ the generator lies on $\mathrm{Re}\,x=\tfrac12$ and on no
uncolorable curve (145 rays, 462 pairs, 30 triads, as in the table).  The row $(1+i)/2$ of
Table~\ref{tab:conj} is the Gaussian alphabet $\{0,\pm1,\pm i,\pm(1+i)\}$ of
Table~\ref{tab:catalogue}, after multiplying every letter by $1+i$.

Fourth, the certified minima of Table~\ref{tab:conj} are 31, 33, 43 and 45.  The alphabet with
the conjugate does not give a KS set smaller than that of Conway and Kochen at any point where
the computation finished, and among the certified cases 31 is reached only by the integer pool.
The statement is incomplete at $x=\pm1\pm i$, where the certifier did not terminate, and at the
class marked 32--33, where the minimum is 32 or 33.

Fifth, like Theorem~\ref{thm:twosymbol} this concerns raw pools.  Completed pools are not
covered.

The script asserts agreement with the number-field constructions of
Section~\ref{sec:catalogue} at seven points, at a generic point and on the circle $|x|^2=2$;
we also rebuilt six further cases with those constructions and found the same counts.  The
graphs of the smallest sets are identified by \texttt{ks\_conjugate\_closed\_identify.py}.

\subsection{The circle $|x|^2=2$}\label{sec:circle}

Earlier versions of this paper met the conjugate-closed alphabet at one point,
$x=(1+\sqrt{-7})/2$, found the pool uncolorable, and attributed that to the arithmetic of the
ring of integers of $\Q(\sqrt{-7})$.  It is a property of the whole circle, and it has a proof
that does not go through the case analysis of Theorem~\ref{thm:conj}.

\begin{theorem}\label{thm:circle}
Let $W$ be the list of 45 words in Table~\ref{tab:words}.
\begin{enumerate}
\item[(i)] For every $x\in\C$ with $|x|^2=2$ the rays of $W$ form a KS-uncolorable set.
Consequently $P(x)$ is uncolorable for every such~$x$.
\item[(ii)] The orthogonalities among the 145 words that hold for every $x\in\C$ form a colorable
configuration (366 pairs, 22 triads).  Hence $P(x)$ is colorable for every $x$ outside a finite
union of proper real algebraic subsets of~$\C$.
\item[(iii)] No list of fewer than 45 of the 145 words is uncolorable by virtue of the
orthogonalities that hold along the whole circle.  Consequently no fixed list of fewer than 45
words over $\mathcal A(x)$ is uncolorable at every point of the circle.
\end{enumerate}
Part~(iii) concerns fixed lists of words over this alphabet.  It says nothing about other
parametrized families of KS sets.
\end{theorem}

\begin{proof}
The inner product of two words is a polynomial in $x$ and $x^{*}$ with integer coefficients, of
degree at most~2.

(i)  On the circle, $xx^{*}=2$ reduces every such polynomial to an integer combination of
$1,x,x^{*},x^2,x^{*2}$.  With $x=\sqrt2e^{i\theta}$ these are constant multiples of
$e^{ik\theta}$, $k=0,\pm1,\pm2$, which are linearly independent functions of~$\theta$.  So an
inner product vanishes at every point of the circle if and only if every reduced coefficient is
zero, a finite check in integers.  It yields 98 pairs among the words of $W$ that are orthogonal
at every point of the circle, forming the 22 triads listed in Appendix~\ref{app:triads}.  A SAT
solver shows that the configuration with these 98 exclusions and 22 triads has no coloring.  Now
fix $x$ on the circle.  If the 45 words give 45 distinct rays, the orthogonality graph at $x$
contains the 98 pairs, and any further pairs or triads only add constraints, so there is no
coloring.  If two words give the same ray, a coloring at $x$ assigns them the same value and so
would also be a coloring of the 45-vertex configuration; there is none.  (Two words that are
orthogonal for every $x$ on the circle cannot coincide at any point of it.)

(ii)  Without the relation, the monomials $x^ax^{*b}$ are linearly independent functions
on~$\C$, so an inner product vanishes for all $x$ iff all its coefficients vanish.  This yields
366 pairs and 22 triads on the 145 words, and a SAT solver finds a coloring.  Each remaining
inner product is a nonzero polynomial in $x,x^{*}$ and vanishes only on a proper real algebraic
subset.  For each pair of the 145 words at least one $2\times2$ minor of the pair is a nonzero
polynomial, so the set of $x$ at which the two words give the same ray is also a proper real
algebraic subset.  Away from those finitely many subsets the 145 rays are distinct, no further
pair is orthogonal, and the coloring found is a coloring of $P(x)$.

(iii)  The 145 words with the pairs and triads that hold along the whole circle (378 pairs, 34
triads) form a finite configuration, and the certifier of Section~\ref{sec:methods} returns 45 as
its minimum, with a certificate of 408 colorable sets.  For the consequence: the exceptional
points of the circle, at which some further inner product vanishes or two words coincide, are
finite in number, since each such condition is a nonzero trigonometric polynomial in~$\theta$.
A list uncolorable at every point of the circle is uncolorable at a non-exceptional point, where
the configuration is exactly the one just described.

The script \texttt{ks\_modulus2\_circle.py} carries out all three computations.  Its
output file, \texttt{results/modulus2\_circle.json}, holds the 145 words, both graphs, the
coloring of part~(ii) and the certificate of part~(iii).
\end{proof}

Part~(i) is also proved in Lean~4: for every complex $x$ with $xx^{*}=2$, the 45 words admit
no coloring (theorem \texttt{Circle45.uncolorable} in the file \texttt{NamedSets.lean}).  The Lean proof follows the argument above.  Inner products are reduced with
$xx^{*}=2$ and nothing else, so every orthogonality it uses holds at every point of the circle;
the orthogonal pairs are recomputed inside Lean; and a case-split refutation of every coloring
is checked by the kernel.  Parts (ii) and~(iii) are not formalized.

\begin{table}[t]
\centering
\caption{The 45 words of Theorem~\ref{thm:circle}.  A bar is a minus sign and $x^{*}$ is the
complex conjugate, so $\bar x^{*}=-x^{*}$.  The list is one minimum set; we have not enumerated
the others.}
\label{tab:words}
\small
\begin{tabular}{rl@{\qquad}rl@{\qquad}rl@{\qquad}rl@{\qquad}rl}
\hline
1 & $0\,0\,1$ & 10 & $1\,0\,1$ & 19 & $1\,1\,x^{*}$ & 28 & $1\,x^{*}\,\bar 1$ & 37 & $x\,\bar 1\,0$ \\
2 & $0\,1\,0$ & 11 & $1\,0\,\bar 1$ & 20 & $1\,1\,\bar x^{*}$ & 29 & $1\,\bar x^{*}\,0$ & 38 & $x\,\bar 1\,1$ \\
3 & $0\,1\,1$ & 12 & $1\,0\,x$ & 21 & $1\,\bar 1\,0$ & 30 & $1\,\bar x^{*}\,1$ & 39 & $x\,\bar 1\,\bar 1$ \\
4 & $0\,1\,\bar 1$ & 13 & $1\,0\,\bar x$ & 22 & $1\,\bar 1\,x$ & 31 & $1\,\bar x^{*}\,\bar 1$ & 40 & $x^{*}\,0\,1$ \\
5 & $0\,1\,x$ & 14 & $1\,0\,x^{*}$ & 23 & $1\,\bar 1\,\bar x$ & 32 & $x\,0\,1$ & 41 & $x^{*}\,0\,\bar 1$ \\
6 & $0\,1\,\bar x$ & 15 & $1\,0\,\bar x^{*}$ & 24 & $1\,\bar 1\,x^{*}$ & 33 & $x\,0\,\bar 1$ & 42 & $x^{*}\,1\,1$ \\
7 & $0\,x^{*}\,1$ & 16 & $1\,1\,0$ & 25 & $1\,\bar 1\,\bar x^{*}$ & 34 & $x\,1\,0$ & 43 & $x^{*}\,1\,\bar 1$ \\
8 & $0\,x^{*}\,\bar 1$ & 17 & $1\,1\,x$ & 26 & $1\,x^{*}\,0$ & 35 & $x\,1\,1$ & 44 & $x^{*}\,\bar 1\,1$ \\
9 & $1\,0\,0$ & 18 & $1\,1\,\bar x$ & 27 & $1\,x^{*}\,1$ & 36 & $x\,1\,\bar 1$ & 45 & $x^{*}\,\bar 1\,\bar 1$ \\
\hline
\end{tabular}

\end{table}

Comparing (i) and (ii): of the 378 pairs present along the circle, 366 are present for every
$x$, and the other 12 are all of the form $1+1-x^{*}x=0$, as in
$(x,1,1)\perp(x,\bar1,\bar1)$.  With them come 12 further triads.  These 12 pairs are the
difference between a colorable pool and an uncolorable one, and the relation behind them is the
one Peres's cube uses at $x=\sqrt2$ (Fig.~\ref{fig:generic}).

\begin{corollary}\label{cor:points}
\begin{enumerate}
\item[(a)] At $x=\sqrt2$ and at $x=\sqrt{-2}$ the 45 words give exactly 33 distinct rays.  By
Computed fact~\ref{fact:peres} these are the unique minimum sets of those pools; at $x=\sqrt2$
they are the 33 rays of Peres.
\item[(b)] At $x=(1+\sqrt{-7})/2$, $P(x)$ is the 145-ray pool studied in earlier versions.  It is
uncolorable because $|x|^2=2$.
\item[(c)] The words $(1,1,x)$ and $(1,0,x^{*})$ of $W$ have squared overlap
$|1+x^{*2}|^2/12=(5+4\cos2\theta)/12$, which is not constant.  The sets of part~(i) are therefore
not all unitarily equivalent: the multiset of squared overlaps between rays is a unitary
invariant that does not depend on labelling, and it takes infinitely many values along the
circle.
\item[(d)] At $x=(1+\sqrt{-7})/2$ the 45 words give 45 distinct rays with 114 orthogonal pairs
and 26 triads, and this set contains a 43-ray KS subset; 43 is the certified minimum within it.
\end{enumerate}
\end{corollary}

\subsection{The point $x=(1+\sqrt{-7})/2$}

Let $\alpha=(1+\sqrt{-7})/2$, so $|\alpha|^2=2$ and $\alpha+\alpha^{*}=1$.

\begin{fact}[The $\Q(\sqrt{-7})$ pool]\label{fact:heegner}
$P(\alpha)$ has 145 rays, 522 orthogonal pairs and 42 triads, against 378 and 34 at a generic
point of the circle.  Its minimum uncolorable subsets have 43 rays, certified.  There are exactly
six of them, in two orbits under signed permutations of the coordinates, each with 23 bases; the
one we use has 107 pairs and degree sequence $4^{16}5^{20}6^{4}8^{2}10^{1}$.  The completed pool
has 685 rays, 1746 pairs and 582 triads; its minimum is \heegnercompletedmin.
\end{fact}

So the additional relation at this point is worth two rays: 43 against 45
(Corollary~\ref{cor:points}d).  It is the pools $P(x)$ that form a continuous family through
this point.  Theorem~\ref{thm:circle} proves that $P(\alpha)$ is uncolorable; that its minimum is
43 is a separate computation.  Both halves of that minimum are also proved in Lean
(\texttt{KSMin.HeegnerPool.minimum\_43}, Section~\ref{sec:methods}): some 43 of the 145 rays are
uncolorable, and every set of fewer than 43 of them is colorable.

\subsection{The circle $|x|^2=\tfrac12$: the family of Gould and Aravind}\label{sec:halfcircle}

On $|x|^2=\tfrac12$ the smallest uncolorable set found has 33 rays and the orthogonality graph
of Peres's set, at a generic point and at every special point (Table~\ref{tab:conj}).  So this
circle carries a one-parameter family of 33-ray sets with Peres's graph, joining Peres's set
($x=1/\sqrt2$) to the $\sqrt{-2}$ set ($x=i/\sqrt2$).  It is the family of Gould and
Aravind~\cite{GouldAravind2010}, who exhibited 33-ray sets with the graph of Peres's set,
depending on three phases and containing the sets of Peres and of Penrose; Flores
Gordillo~\cite{FloresGordillo2026} reports that after removing diagonal unitaries their family
is a circle.  The following are exact symbolic computations
(\texttt{ks\_half\_circle\_family.py exact}), writing $x=e^{i\theta}/\sqrt2$.
\begin{enumerate}[label=(\alph*)]
\item Exactly 72 inner products of the 33 words vanish identically in~$x$, no other inner
product vanishes anywhere on the circle, and no two words become proportional.  So the graph is
Peres's along the whole circle, not only at the special points.
\item Substituting $a=b=1$ and $c=1/x^{*}$ in Table~3 of Ref.~\cite{GouldAravind2010} gives our
33 words, ray for ray; and for general $a,b$ (of modulus~1) the diagonal unitary
$\mathrm{diag}(b,a,1)$ carries their set with parameters $(a,b,c)$ to their set with parameters
$(1,1,ab^{*}c)$.  So up to unitary equivalence their three-parameter family is our circle, in
agreement with Flores Gordillo's reduction.
\item The maps $x\mapsto-x$ and $x\mapsto x^{*}$ (the latter followed by exchanging the first
two coordinates) carry the set of 33 rays to itself, so the members at $\theta$, $\theta+\pi$
and $-\theta$ are unitarily equivalent.
\item For unit vectors, $\sum_{i<j}|\braket{v_i}{v_j}|^4=
\tfrac{37493}{432}+\tfrac{17}{2}\cos2\theta$, which is strictly decreasing on $[0,\pi/2]$.
Hence the members with $0\le\theta\le\pi/2$ are pairwise unitarily inequivalent, and every
member is equivalent to one of them.
\end{enumerate}
Gould and Aravind identify Penrose's 33 rays with the parameters
$a=-i$, $b=-1$, $c=-\sqrt2$, which we have not re-derived.  These give $x=-i/\sqrt2$, so on that
identification Penrose's set is unitarily equivalent to the $\sqrt{-2}$ set; consistently, the
overlap of rays 9 and~14 is $(2-\sqrt2)/4$ at $\theta=0$ and $\sqrt6/4$ at $\theta=\pi/2$, the
values they give for Peres's and Penrose's sets.

This family is different from the one in Theorem~\ref{thm:circle}.  Its members have 33 rays
and a constant graph.  The family of Theorem~\ref{thm:circle} has 45 rays at a generic point and
cannot contain a 33-ray KS set there, since the minimum at a generic point of that circle is 45
(Theorem~\ref{thm:conj}).  The two families meet in Peres's set: $P(\sqrt2)=P(1/\sqrt2)$, a
49-ray pool whose unique minimum subset is Peres's set.  We
have not found the family of Theorem~\ref{thm:circle} in the literature, but we have not made a
systematic search and make no claim of priority.

\subsection{Rigidity}\label{sec:rigidity}

For the minimum sets and uncolorable witnesses named below we computed numerically the number
of infinitesimal flexes in $\C^3$: the
dimension of the null space of the Jacobian of the equations $|v_i|^2=1$,
$\braket{v_i}{v_j}=0$ (one complex equation for each orthogonal pair), minus the $n+8$
dimensions accounted for by a common unitary and a phase on each vector.  In every case the
smallest retained singular value exceeds $0.2$ and the largest discarded one is below
$2\times10^{-15}$.  The counts (\texttt{ks\_rigidity\_v10.py}, and
\texttt{ks\_half\_circle\_family.py} for the family on $|x|^2=\tfrac12$) are: no flex for $\CK$,
the Eisenstein 33-ray set, the 43-ray set, the two 47-ray supergolden sets (the minimum set of
the raw three-symbol pool, and the witness in the completed pool of $\{0,\pm1,\pm\psi\}$), the
51- and 60-ray witnesses of Section~\ref{sec:catalogue}, which are not known to be minimum sets,
and the 45 words at $x=(1+\sqrt{-7})/2$ with the 16 additional
orthogonalities of that point imposed; one flex for the Peres and $\sqrt{-2}$ sets and for each
of nine sampled members of the family on $|x|^2=\tfrac12$; two flexes for the 45 words at each
of the two generic phases tried ($\theta=0.3$ and $\theta=1.0$).  Zero flexes, if confirmed by
an exact rank computation, means the set is locally rigid.  It does not by itself mean that the
realization is unique up to unitary equivalence; for $\CK$ that stronger statement is the
theorem of Trandafir and Cabello~\cite{TrandafirCabello2025}, and for the Eisenstein set it is
asserted by Cabello~\cite{Cabello2025simplest}.  The single flex of the Peres set is the circle
of Section~\ref{sec:halfcircle}, a genuine deformation and not merely an infinitesimal one.
(Versions 1--7 claimed the opposite on the strength of a second-order computation with an
indexing error, found by Flores Gordillo and corrected in version~8.)  Of the two flexes of the
45 words, one is motion along the circle (Corollary~\ref{cor:points}c); we have not identified
the other.

\section{Roots of unity}\label{sec:roots}

Let $\zeta_n=e^{2\pi i/n}$, let $\mu_n$ be the group of $n$th roots of unity, and put
$S_n=S(\{0\}\cup\mu_n)$.  The case $n=6$ is the alphabet $\{0,\pm1,\pm\omega,\pm\omega^2\}$ with
$\omega=e^{2\pi i/3}$, which contains the 33-ray set that Cabello calls the simplest KS
set~\cite{Cabello2025simplest}.

\begin{lemma}\label{lem:sums}
A sum of two elements of $\mu_n$ can vanish if and only if $2\mid n$.  A sum of three can vanish
if and only if $3\mid n$, and then the three terms are $\beta,\beta\omega,\beta\omega^2$ for some
$\beta\in\mu_n$.
\end{lemma}

\begin{proof}
$\zeta^a+\zeta^b=0$ iff $\zeta^{b-a}=-1$, and $-1\in\mu_n$ iff $2\mid n$.  For three terms divide
by the first: $1+e^{i\theta}+e^{i\phi}=0$.  The imaginary part gives $\phi=-\theta$ or
$\phi=\pi+\theta$; the second makes the real part equal to~1.  With $\phi=-\theta$ the real part
gives $\cos\theta=-\tfrac12$, so the other two terms are $\omega$ and $\omega^2$, which lie in
$\mu_n$ iff $3\mid n$.
\end{proof}

Hence two rays of $S_n$ are orthogonal only if the number of coordinates in which both are
nonzero is~0, or~2 with $2\mid n$, or~3 with $3\mid n$.

\begin{theorem}\label{thm:6n}
$S_n$ is KS-uncolorable if and only if $6\mid n$.
\end{theorem}

\begin{proof}
\emph{If.}  For $6\mid n$ we have $\mu_6\subseteq\mu_n$ and so $S_6\subseteq S_n$.  The pool
$S_6$ is uncolorable (Table~\ref{tab:cyc}; it contains the 33-ray set
of~\cite{Cabello2025simplest}).  Apply Lemma~\ref{lem:restrict}.

\emph{Only if.}  Let $6\nmid n$ and write $e_1,e_2,e_3$ for the axis rays.

\emph{Case 1: $2\nmid n$ and $3\nmid n$.}  Orthogonal rays have disjoint supports, so the only
triad is $\{e_1,e_2,e_3\}$.  Give $e_1$ the value~1 and every other ray~0.

\emph{Case 2: $2\mid n$, $3\nmid n$.}  Two rays with three common nonzero coordinates are never
orthogonal.  No triad contains an all-nonzero ray~$a$: each of its two partners would share
exactly two nonzero coordinates with~$a$ and so have support of size~2; two such supports meet,
so the partners, being orthogonal to each other, would have the same support; and then $a$ would
be orthogonal to a basis of that coordinate plane and vanish on it.  Rays whose supports are
different two-element sets share one coordinate and are not orthogonal.  So every triad is
$\{e_1,e_2,e_3\}$ or $\{e_k,v,v'\}$ with $v,v'$ in the coordinate plane orthogonal to~$e_k$.  In
such a plane the rays are $(1,\zeta^p)$, and $(1,\zeta^p)\perp(1,\zeta^q)$ iff $q-p\equiv n/2$: a
perfect matching.  Give $e_1$ the value~1 and $e_2,e_3$ the value~0; give~0 to every ray of the
plane orthogonal to $e_1$; in each of the other two coordinate planes give~1 to one ray of each
matched pair and~0 to the other; give~0 to every all-nonzero ray.  Rays in different coordinate
planes are not orthogonal, and no ray with value~1 in those two planes is orthogonal to~$e_1$, so
no two rays with value~1 are orthogonal, and each triad has exactly one~1.

\emph{Case 3: $3\mid n$, $2\nmid n$.}  Orthogonality needs common support of size 0 or~3.  A ray
with exactly one zero coordinate is therefore orthogonal only to the axis ray on that coordinate,
and lies in no triad.  For an all-nonzero ray normalise $v=(1,a,b)$.  By Lemma~\ref{lem:sums} an
all-nonzero $w$ is orthogonal to $v$ iff $w=w_{\pm}=(1,a\omega^{\pm1},b\omega^{\mp1})$, and
$\braket{w_+}{w_-}=1+\omega^{-2}+\omega^{2}=0$.  So the all-nonzero rays fall into disjoint
triads, with no other orthogonalities among them and none to a ray with a zero coordinate.  Give
the value~1 to one ray of each such triad and to $e_1$, and~0 to everything else.
\end{proof}

For $n\ge3$ the theorem is also proved in Lean~4 with the mathlib library, in the file
\texttt{lean/Cyclotomic6n.lean} of the repository, with no unproved steps and no axioms
beyond the standard three.  There the uncolorability of $S_6$ is checked rather than cited: 33 rays of $S_6$ are
written over $\Z[\omega]$, their orthogonalities are recomputed in that ring, and a case-split
refutation of every coloring is checked by the Lean kernel.  For $n=1,2$ the pool lies in
$S(\{0,\pm1\})$, which is colorable.

The two primes do different work.  Divisibility by~3 creates triads: at $n=9$ and $n=15$ the pool
has 28 and 76 of them and is colorable because they are pairwise disjoint.  Divisibility by~2
creates the two-term cancellations that couple triads.  Only together do the constraints
interlock.  The fields $\Q(\zeta_3)$ and $\Q(\zeta_6)$ coincide, while $S_3$ is colorable and
$S_6$ is not: Theorem~\ref{thm:6n} is a statement about alphabets, not about fields, and it has
no formulation in terms of the Galois group of the field (the one offered in version~9 was
false).

\begin{table}[t]
\centering
\caption{Cyclotomic pools, exact arithmetic (exponents modulo $n$), colorability by SAT.  The
check was run for all $2\le n\le30$ and agrees with Theorem~\ref{thm:6n} in every case, as do the
three structural statements used in Cases 1--3; a selection is shown.}
\label{tab:cyc}
\begin{tabular}{rrrrc@{\qquad\qquad}rrrrc}
\toprule
$n$ & Rays & Pairs & Triads & Colorable & $n$ & Rays & Pairs & Triads & Colorable\\
\midrule
3 & 21 & 21 & 4 & yes & 12 & 183 & 633 & 67 & \textbf{no}\\
4 & 31 & 69 & 7 & yes & 15 & 273 & 273 & 76 & yes\\
5 & 43 & 18 & 1 & yes & 16 & 307 & 843 & 25 & yes\\
6 & 57 & 174 & 22 & \textbf{no} & 18 & 381 & 1380 & 136 & \textbf{no}\\
9 & 111 & 111 & 28 & yes & 20 & 463 & 1293 & 31 & yes\\
\bottomrule
\end{tabular}
\end{table}

\begin{fact}[Eisenstein pool]\label{fact:eisenstein}
The 57-ray pool $S_6$ has 174 orthogonal pairs and 22 triads.  Its minimum uncolorable subsets
have 33 rays, certified, and there are exactly three of them, each with 78 pairs and 14 bases.
The pool is not closed under completion: $(1,1,0)$ and $(1,\bar1,1)$ are orthogonal and their
third ray $(1,\bar1,\bar2)$ is not in $S_6$.  The completed pool has 165 rays, 390 pairs and 130
triads, and its minimum is also 33, certified.
\end{fact}

Since every $S_n$ with $6\mid n$ contains $S_6$, its minimum is at most~33.  We have not
certified the minimum for $n>6$.

\begin{remark}
Cortez, Morales and Reyes~\cite{CortezMoralesReyes2022} showed that $\Z[1/6]$ is the smallest
ring $\Z[1/N]$ over which the partial ring of symmetric $3\times3$ matrices admits no algebraic
hidden state.  The number~6 appears there
and in Theorem~\ref{thm:6n}, but the two statements concern different objects (matrix entries of
projectors in one, coordinates of rays in the other) and we know of no map from one to the
other.  Version~9 presented this as a quantitative connection; it is a coincidence until someone
proves otherwise.
\end{remark}

\section{Other pools}\label{sec:catalogue}

Table~\ref{tab:catalogue} collects the named uncolorable pools, including those that lie outside
the families of Sections~\ref{sec:real}--\ref{sec:roots}, with the standard of proof for each
minimum.  ``Certified'' means the procedure of Section~\ref{sec:methods} terminated;
``$\le$'' means an uncolorable subset of that size is in hand and nothing smaller was excluded;
$^{\mathrm L}$ marks a minimum that is also proved in Lean.  A raw pool is contained in its
completion, so a minimum set of the raw pool bounds the completed pool from above
(Lemma~\ref{lem:restrict}); it need not be a minimum set of it.

\begin{table}[t]
\centering
\caption{Uncolorable pools.  ``raw'': all words over the alphabet; ``compl.'': closure under
third rays of orthogonal pairs.  Letters: $\omega=e^{2\pi i/3}$, $\alpha=(1+\sqrt{-7})/2$,
$\varphi=(1+\sqrt5)/2$, $\psi^3=\psi^2+1$ ($\psi\approx1.4656$), $c=\sqrt[3]{2}$.  Bases: number
of triads inside the minimum (or best) set found.  $^{\mathrm L}$: minimum also proved in Lean
(Section~\ref{sec:methods}).}
\label{tab:catalogue}
\small
\begin{tabular}{llrrrll}
\toprule
Alphabet & Pool & Rays & Pairs & Triads & Minimum KS subset & Bases \\
\midrule
$\{0,\pm1,\pm2\}$ & raw & 49 & 138 & 26 & 31, certified$^{\mathrm L}$ & 17 \\
 & compl. & 109 & 258 & 86 & 31, certified & 17 \\
$\{0,\pm1,\pm\sqrt2\}$ & raw & 49 & 120 & 16 & 33, certified$^{\mathrm L}$ & 16 \\
 & compl. & 145 & 336 & 112 & 33, certified & 20 \\
$\{0,\pm1,\pm\sqrt{-2}\}$ & raw & 49 & 120 & 16 & 33, certified$^{\mathrm L}$ & 16 \\
 & compl. & 145 & 336 & 112 & 33, certified & 20 \\
$\{0,\pm1,\pm\omega,\pm\omega^2\}$ & raw & 57 & 174 & 22 & 33, certified$^{\mathrm L}$ & 14 \\
 & compl. & 165 & 390 & 130 & 33, certified & 14 \\
$\{0,\pm1,\pm i,\pm(1+i)\}$ & raw & 127 & 465 & 51 & 33, certified & 16 \\
 & compl. & 583 & 1521 & 507 & $\le33$ & \\
$\{0,\pm1,\pm\alpha,\pm\alpha^{*}\}$ & raw & 145 & 522 & 42 & 43, certified$^{\mathrm L}$ & 23 \\
 & compl. & 685 & 1746 & 582 & \heegnercompletedmin & \\
$\{0,\pm1,\pm\psi,\pm\psi^2\}$ & raw & 109 & 348 & 40 & 47, certified$^{\mathrm L}$ & 26 \\
 & compl. & \superfourcompleted & \superfourmin \\
$\{0,\pm1,\pm\psi\}$ & compl. & \supercompleted & \supermin \\
$\{0,\pm1,\pm\varphi\}$ & compl. & 205 & 498 & 166 & \goldenmin & \goldenbases \\
$\{0,\pm1,\pm c,\pm c^2\}$ & compl. & \cubiccompleted & \cubicmin \\
\bottomrule
\end{tabular}
\end{table}

\paragraph{Completion.}  The last three rows come from a different construction.  The raw pools
of $\{0,\pm1,\pm\varphi\}$ (49 rays, 10 triads), $\{0,\pm1,\pm\psi\}$ (49 rays, 10 triads) and
$\{0,\pm1,\pm c,\pm c^2\}$ (109 rays, 28 triads) are colorable; uncolorability appears only
after completion, which adds rays whose coordinates are not letters.  The completed pools of
$\{0,\pm1,\pm\sqrt3\}$ and $\{0,\pm1,\pm\sqrt5\}$ (133 rays, 282 pairs, 94 triads each) remain
colorable.  Since $\sqrt5$ and $\varphi$ generate the same field, this is the plainest evidence
that the alphabet, and not the field, is the operative choice.  We have no account of why
completion suffices in one case and not in the other.  (Earlier versions said that completion
``generates the reciprocal $1/\varphi$''; that explains nothing, since
$[\varphi:1:0]=[1:1/\varphi:0]$ is already in the raw pool.)  For these pools and for the other
completions marked ``$\le$'' the certifier did not terminate, and the entries are upper bounds:
a 51-ray subset for the golden ratio (118 pairs, 28 bases; coordinates in
\texttt{golden\_51\_rays.txt}; version~9 claimed a certified minimum of 52, which is wrong); a
60-ray subset for the cube root of two (144 pairs, 27 bases); and 47 for both supergolden
completions.

\paragraph{The supergolden alphabet.}  Let $\psi$ be the real root of $\psi^3=\psi^2+1$.  The
raw pool of $\{0,\pm1,\pm\psi,\pm\psi^2\}$ is uncolorable without any completion, with certified
minimum 47 (110 pairs, 26 bases); it is one of the eight points of Table~\ref{tab:threesymbol}.
A 47-ray set over this field was first communicated to us by J.~L.~Gonzalez, extracted from the
completed pool of $\{0,\pm1,\pm\psi\}$.  It refuted the classification of versions 1--8, because
no product of two letters of $\{0,\pm1,\pm\psi\}$ has squared modulus~2 and the field contains no roots of unity other
than $\pm1$: the relation at work is the minimal polynomial of $\psi$ itself.  The conclusion we
now draw is that any relation among the letters that makes enough inner products vanish may do.
Theorem~\ref{thm:threesymbol} lists all such relations for three real letters; for larger
alphabets we have no finite classification.

\begin{table}[t]
\centering
\caption{The minimum sets of the certified raw pools and the 33-ray set of Sch\"utte (Conway
and Kochen), which is critical but not a minimum set of its pool.  Degree sequences are written
$d^{m}$ for $m$ rays orthogonal to $d$ others.  ``Indisp.'': number of bases whose ``exactly
one'' rule cannot be dropped (all orthogonal-pair exclusions kept) without the set becoming
colorable.  Every set is ray-critical.}
\label{tab:profiles}
\small
\begin{tabular}{lrrrll}
\toprule
Set & Rays & Pairs & Bases & Degree sequence & Indisp. \\
\midrule
$\CK$ & 31 & 71 & 17 & $3^{4}4^{14}5^{8}6^{3}8^{2}$ & 17 of 17 \\
Sch\"utte / Conway--Kochen 33 & 33 & 76 & 20 & $3^{8}4^{13}6^{10}8^{2}$ & 20 of 20 \\
Peres 33 & 33 & 72 & 16 & $4^{30}8^{3}$ & 13 of 16 \\
$\sqrt{-2}$, 33 & 33 & 72 & 16 & $4^{30}8^{3}$ & 13 of 16 \\
Eisenstein 33 & 33 & 78 & 14 & $4^{18}5^{12}8^{3}$ & 13 of 14 \\
$\Q(\sqrt{-7})$, 43 & 43 & 107 & 23 & $4^{16}5^{20}6^{4}8^{2}10^{1}$ & 18 of 23 \\
Supergolden 47 & 47 & 110 & 26 & $3^{14}4^{12}5^{8}6^{6}7^{4}8^{1}9^{2}$ & 26 of 26 \\
\bottomrule
\end{tabular}
\end{table}

Table~\ref{tab:profiles} gives the basic invariants of the minimum sets.  The three 33-ray sets
of Conway--Kochen, Peres and the Eisenstein alphabet have pairwise non-isomorphic orthogonality
graphs (VF2, and visibly from the degree sequences).  The Eisenstein set is the 33-ray set of
Cabello~\cite{Cabello2025simplest}, or one of its three images in the pool
(Computed fact~\ref{fact:eisenstein}); we have matched the invariants he reports (33 rays, 14
bases), and the explicit bijection has not been recorded in the repository.

Whether the 43-ray and 47-ray sets, or the 45-ray family of Theorem~\ref{thm:circle}, appear in
the published catalogues we do not know.  A proper comparison against the large automated
catalogues of three-dimensional KS sets has not been done, and we make no claim that these
configurations are new.

\section{Open questions}\label{sec:limitations}

The questions fall into two kinds.  The first three are pending further resources: the method of
Section~\ref{sec:methods} applies to them as it stands, and the runs either did not finish
within the budgets we gave them or have not yet been attempted.  They need longer runs or faster
solvers, not new ideas.  The remaining items are of other kinds: conceptual explanations we do
not have, extensions we have not attempted, a literature comparison, and further formalization.

\paragraph{Pending further resources.}
\begin{enumerate}
\item \emph{Two minima.}  Theorem~\ref{thm:conj} leaves the minimum undetermined at two
classes of special points: it is 32 or 33 at $x=\zeta/\sqrt2$ ($\zeta=e^{i\pi/3}$),
and at $x=1+i$ it lies between the general lower bound 24 and~45.  A set of fewer than 31 rays
at $x=1+i$ is not excluded.
\item \emph{Uncertified minima.}  The entries of Table~\ref{tab:catalogue} marked ``$\le$''.
\item \emph{All minimum sets on the circle.}  We exhibit one 45-word minimum set; the others have
not been enumerated.
\end{enumerate}

\paragraph{Open.}
\begin{enumerate}[resume]
\item \emph{The two circles.}  The family of 33-ray sets on $|x|^2=\tfrac12$ is that of Gould
and Aravind (Section~\ref{sec:halfcircle}).  The minima at generic points, 33 on that circle and
45 on $|x|^2=2$, are both results of the certifier.  We have no argument by hand for either, and
no explanation of why the two circles differ.  (The class $\zeta/\sqrt2$ on the inner circle
remains undetermined, item~1.)
\item \emph{Complex three-symbol alphabets.}  Theorem~\ref{thm:threesymbol} settles real
$\{0,\pm1,\pm x,\pm y\}$.  For complex $x,y$ the relations are again finitely many, but the
parameter space has four real dimensions and we have not attempted the enumeration.  Could a
complex three-symbol pool contain a KS set of fewer than 31 rays?
\item \emph{The second flex.}  Two infinitesimal flexes of the 45 words were found numerically
at both phases tested (Section~\ref{sec:rigidity}).  Is the second one a genuine deformation, leading off the
family of Theorem~\ref{thm:circle}?
\item \emph{Completion.}  Is there a finiteness criterion for completed pools?  Why does
completion make the golden-ratio and $\sqrt[3]2$ pools uncolorable but not those of $\sqrt3$ and
$\sqrt5$?
\item \emph{The gap.}  This paper does not close the gap between the lower bound 24 and
the 31 rays of Conway and Kochen.  Theorems~\ref{thm:twosymbol} and~\ref{thm:threesymbol} exclude a
smaller set with a two-symbol or a real three-symbol alphabet and nothing more; the other minima are certified within named finite pools
only.  Earlier
versions suggested that the alphabet analysis could prune the hypergraph searches
of~\cite{LiBrightGanesh2024,KirchwegerPeitlSzeider2023}; we see no basis for that, since a
hypothetical smaller set need not have a small alphabet.
\item \emph{Catalogue comparison.}  Whether the 43-, 45- and 47-ray sets appear in the automated
catalogues of three-dimensional KS sets.  (This needs reading, not computation, and has not been
done.)
\item \emph{Formalization.}  The three continuous-parameter classifications
(Theorems~\ref{thm:twosymbol}, \ref{thm:threesymbol} and~\ref{thm:conj}) are not formalized:
Lean checks Theorem~\ref{thm:6n} for $n\ge3$ and the uncolorability and minimality statements
listed in Section~\ref{sec:methods}, not the reductions to finitely many cases.  Nor does it check that each pool list is complete for its alphabet, or
part~(iii) of Theorem~\ref{thm:circle}.
\end{enumerate}

\section{Conclusion}\label{sec:conclusion}

Reading the coordinates off Peres's two cubes suggests a one-parameter question: for which $x$
does the alphabet $\{0,\pm1,\pm x\}$ generate an uncolorable pool?  The answer is now complete
(Theorem~\ref{thm:twosymbol}).  Up to the sign and inversion of $x$, which do not change the pool, it is the two classical cases
$x=\sqrt2$ and $x=2$ and nothing else for real $x$; for complex $x$ it adds only $\sqrt{-2}$ and
the primitive third and sixth roots of unity; and in no case is there a KS set of fewer than 31 rays.  One unit and one
further number along each axis cannot improve on Conway and Kochen, and neither can two further
real numbers (Theorem~\ref{thm:threesymbol}).  The more interesting
structure appears when the alphabet is enlarged.  With the conjugate
included, the uncolorable values of $x$ are no longer isolated: they are the two circles
$|x|^2=2$ and $|x|^2=\tfrac12$, four lines through the origin (without the origin), and the points
$\pm2,\pm\tfrac12$ (Theorem~\ref{thm:conj}).  On the circle $|x|^2=2$ twelve orthogonalities of the form $1+1-|x|^2=0$ make the difference; the Peres set
and its imaginary twin are the two degenerate members, up to sign, of a 45-ray family; and the 43-ray set over
$\Q(\sqrt{-7})$ sits at a point of the circle where one more relation holds.  For the
roots-of-unity alphabets the answer is exact as well: uncolorable if and only if $6\mid n$, with
the primes 2 and 3 doing different jobs.

The earlier versions of this paper claimed more than this and were wrong in ways that matter.
We have tried to make this version checkable instead: every number is produced by a named script
in exact arithmetic or is labelled numerical, every minimum is either certified by a procedure
whose soundness is argued in two paragraphs or reported as a bound, six of the minima and
Theorem~\ref{thm:6n} for $n\ge3$ are checked by the Lean kernel, and the withdrawn claims are listed in
Appendix~\ref{sec:corrections}.  Code, data and the scripts behind each table are listed in
Appendix~\ref{app:repro}; Appendix~\ref{app:produced} says how the paper was produced, including
the part played by language models.

\subsection*{Acknowledgments}

The author thanks Ad\'an Cabello for pointing out the connection between the Eisenstein alphabet
and his 33-ray set; Manuel Flores Gordillo for finding the error in the rigidity computation of
versions 1--7; Jorge Lucas Gonzalez for the supergolden construction and an exact verifier;
Markus Frembs for the trace-zero relation and for engaging with the earlier classification; and
Curtis Bright for a bibliographic correction.  The author accepts responsibility for the content,
including the errors of the earlier versions.

\appendix
\section{Corrections to versions 1--9}\label{sec:corrections}

This appendix lists the principal claims of versions 1--9 that were withdrawn or changed in
version~10, so that a reader of an earlier version can see which statements should no longer be
relied on.  It is not guaranteed
complete: where an earlier version and this one disagree, this one is intended.  Items marked \emph{wrong} are
contradicted by a computation or an argument in this version; items marked \emph{not
re-verified} may be right but were not recomputed with the exact methods of
Section~\ref{sec:methods} and are therefore no longer claimed.

\paragraph{Withdrawn as wrong.}
\begin{enumerate}[label=C\arabic*.,leftmargin=2.6em]
\item \emph{Golden-ratio minimum 52.}  A 51-ray uncolorable subset exists
(Section~\ref{sec:catalogue}).  Version~9 also reported a 58-ray set for the completed pool of
$\{0,\pm1,\pm\sqrt[3]2,\pm\sqrt[3]4\}$, from the same unsound procedure.  We have not reproduced
a 58-ray set and do not claim one; that figure is unsupported and withdrawn, not disproved.  The
smallest set we have is 60 rays.
\item \emph{``OCUS-certified'' minima.}  The certification code was unsound
(Section~\ref{sec:methods}).  All minima in this version were recomputed with a different
procedure.  The values 31, 33, 33, 33 and 43 for the integer, Peres, $\sqrt{-2}$, Eisenstein and
$\Q(\sqrt{-7})$ raw pools survive; 52 does not.
\item \emph{``Islands are isolated; no continuous path connects one KS-supporting alphabet to
another.''}  The alphabet $\{0,\pm1,\pm x,\pm x^{*}\}$ is uncolorable along the whole circle
$|x|^2=2$, which joins the Peres, $\sqrt{-2}$ and $\Q(\sqrt{-7})$ alphabets
(Theorem~\ref{thm:circle}).  The sweep of 164 angles offered as evidence could not have
established isolation in any case.
\item \emph{The norm bound}: ``generators with $|x|^2\ge3$ that are not roots of unity produce no
KS sets.''  It fails for $x=2$, the alphabet of $\CK$; $|x|^2$ is not an invariant of the pool
($x$ and $1/x$ give the same rays); and version~9 itself reported uncolorable pools from
generators of squared modulus 4, 5 and~6.
\item \emph{The two- and three-mechanism classifications} (``modulus-2'', ``phase'',
``minimal-polynomial'', ``indirect modulus-2'').  The statement that an uncolorable pool requires
a cancellation identity is true by definition.  The list of mechanisms was extended after each
counterexample and made no prediction.  What replaces it is
Theorems~\ref{thm:twosymbol}, \ref{thm:threesymbol}, \ref{thm:conj}, \ref{thm:circle}
and~\ref{thm:6n} and Table~\ref{tab:zerosums}.
\item \emph{Completeness of the table of three-term zero-sums.}  The earlier table was compiled
by inspection, its justification miscounted the cases, it omitted two-term sums, and it described
solution sets that are circles or lines as isolated points.  Table~\ref{tab:zerosums} is the
complete enumeration.  The accompanying claim that an SMT solver had certified completeness
overstated what was run.
\item \emph{``Universal merge saturation''} and the conjecture built on it.  Identifying two
non-orthogonal rays of any uncolorable configuration gives an uncolorable configuration: a
coloring of the quotient pulls back to the original by giving the two rays the same value.  The
3{,}756 SAT calls verified this lemma; nothing about minimal KS sets follows.  (The printed
definition also described deletion of a ray, not identification.)
\item \emph{The Galois-theoretic reformulation.}  The statement that $(\Z/n\Z)^{*}$ surjects onto
$\Z/6\Z$ if and only if $6\mid n$ is false in both directions ($n=6$, $n=7$); the conductor of
$\Q(\zeta_6)$ is 3, not 6; the norm characterization fails for real quadratic generators; and the
list of ``KS-supporting'' imaginary quadratic fields contradicted the paper's own tables.  The
parallel drawn with Gleason's theorem and with the matrix trace is withdrawn with it.
\item \emph{The Hoffman-bound explanation of CSW advantages}, the ``auxiliary ray'' mechanism and
the description of the Peres pool as disconnected (it is connected).
\item \emph{``$|S_A|+|S_B|\ge|\mathcal B|$ in every case''} for bipartite perfect strategies
(the $\Q(\sqrt{-7})$ row of version~9's own table violated it), and the
claim that every Bell scenario obtainable from $\CK$ is obtainable more cheaply from the
Eisenstein set, which rested on 200 heuristic trials.
\item \emph{The tables of bipartite-strategy costs and Cabello--Severini--Winter quantities.}
The bipartite costs came from a search that was not exhaustive.  Both tables are omitted from
this version, which makes no claim about either.  Recomputed values (exact minimum costs for
nine sets; independence number, Lov\'asz number and fractional packing number for nine sets and
six pools) are in the repository as \texttt{results/bpqs\_exact\_v10.jsonl} and
\texttt{results/csw\_v10.jsonl}.
\item \emph{``The Peres and $\Z[\sqrt{-2}]$ islands are isomorphic.''}  True of the 33-ray minimum
sets, false of the pools (Computed fact~\ref{fact:pools}).
\item \emph{Counts.}  $\{0,\pm1,\pm i\}$ has 31 rays, 69 pairs, 7 triads, not 25/48/6.  The
completed pools of $\sqrt3$ and $\sqrt5$ have 133 rays, not 145.  The Eisenstein pool is not
closed under completion (165 rays after completion, not 57).  The $\Q(\sqrt{-7})$ pool has 522
orthogonal pairs in total; version~9 said its 76 auxiliary rays ``add'' 522.  The integer pool
has 6 minimum sets in one symmetry class, not 162 in six.
\item \emph{The Hermitian cross-product formula} was printed with a spurious alternating sign;
the third ray of an orthogonal pair $v,w$ is $[(v\times w)^{*}]$.
\item \emph{``All coordinates are exact, with no floating-point approximation''} was not true of
all the scripts cited.  In this version it is true of the pool counts, colorability decisions,
minima and isomorphism tests; the rigidity ranks are numerical.
\item \emph{Further items.}  The claim that enlarged Eisenstein alphabets keep minimum 33 (once
the letters $\pm2$ are admitted the pool contains $\CK$).  The use of ``contextual fraction'' for
a ratio of bases, a term with an established different meaning.  The old certification script
could print ``certified'' after finding a smaller counterexample.\item \emph{Smaller points.}  The 13-ray pool was called too small ``because the minimum known is
31'': 31 is an upper bound, and the relevant fact is the lower bound 24.  Infinitesimal rigidity
was called necessary for global rigidity; it is neither necessary nor sufficient.  The rigidity
theorem of Trandafir and Cabello concerns realizations of the $\CK$ graph and does not imply that
every 31-ray KS set is $\CK$.  The Kirchweger--Peitl--Szeider reference had the wrong venue,
pages and an author's initial.
\end{enumerate}

\paragraph{No longer claimed (not re-verified).}
The survey of imaginary quadratic fields of class number 2--4; the survey of 28 two-generator
alphabets; the quartic alphabets; mixed alphabets with two independent irrationals; the
trigonometric sweep; the icosahedral directions; the quaternionic and octonionic alphabets (for
the latter the notions of ray and orthogonality were never defined); the comparison with the
97-ray closure of Trandafir and Cabello; the $N(S)$ comparison with Cortez, Morales and Reyes;
the ``graph universality'' of $\CK$ across group-orbit constructions; the ``30{,}000 random
hypergraphs''; the identification of the coordinate field of the original 117-ray set of Kochen
and Specker; the statement that the 43-ray and golden-ratio sets are new to the literature.

\section{Reproducibility}\label{app:repro}

Code and data: \url{https://github.com/michaelkernaghan/contextuality}, release tag
\texttt{v11.1}.  Scripts for this
version, with the result files they write under \texttt{results/}:

\begin{itemize}\setlength{\itemsep}{0pt}
\item \texttt{ks\_recertify\_v10.py}: pools of Table~\ref{tab:catalogue}, with certified minima
and witnesses, written to \texttt{recertify\_v10.jsonl}.  It uses the exact fields, pools and
completion of \texttt{ks\_verify\_large\_minima.py} and the certifier
\texttt{ks\_certify\_fast3.py}.
\item \texttt{ks\_tables\_v10.py}: Table~\ref{tab:alphabets}, the survey of imaginary quadratic
alphabets quoted in Section~\ref{sec:conjclosed}, and the three-symbol and completed pools
quoted in the text (\texttt{tables\_v10.jsonl}).
\item \texttt{ks\_zero\_sums\_v10.py}: Table~\ref{tab:zerosums}.
\item \texttt{ks\_two\_symbol\_theorem.py}: Theorem~\ref{thm:twosymbol} and
Table~\ref{tab:twosymbol} (\texttt{two\_symbol\_theorem.json}, with a witness for every
uncolorable case); about five seconds.
\item \texttt{ks\_conjugate\_closed\_theorem.py}: Theorem~\ref{thm:conj} and
Table~\ref{tab:conj}.  The mode \texttt{scan} builds and colors the 68 cases in a few seconds.
The mode \texttt{certify} adds the minima, written to
\texttt{conjugate\_closed\_theorem.json}.  The mode \texttt{one} reruns a single case,
optionally only to exclude sets below a given size; it appends to
\texttt{conjugate\_closed\_reruns.jsonl}.  \texttt{ks\_conjugate\_closed\_identify.py}
identifies the smallest sets.
\item \texttt{ks\_three\_symbol\_real.py}: Theorem~\ref{thm:threesymbol} and
Table~\ref{tab:threesymbol}.  The mode \texttt{scan} finds the curves and special points and
colors the 162 cases in exact arithmetic, in under a minute; \texttt{check} rebuilds every
graph from 60-digit coordinates; \texttt{certify} adds the minima and the runs excluding sets
below~31; \texttt{witness} writes a checked minimum set for every case.  Results are in the
files \texttt{three\_symbol\_real\_*} under \texttt{results/}.  It uses \texttt{ks\_certify\_fast4.py}, the
certifier with the two reductions of Section~\ref{sec:methods}, which the mode \texttt{one4} of
\texttt{ks\_conjugate\_closed\_theorem.py} also uses.
\item \texttt{ks\_verify\_cyclotomic\_theorem.py}: Table~\ref{tab:cyc} and the structural
statements of Theorem~\ref{thm:6n}, $n\le30$.
\item \texttt{lean/}: the Lean proofs (Lean 4.30, mathlib; \texttt{lake build} checks all
three files in about ten minutes).  \texttt{Cyclotomic6n.lean}: Theorem~\ref{thm:6n} for
$n\ge3$.  \texttt{NamedSets.lean}: uncolorability of the 43- and 47-ray sets and of the 45
words for every $x$ on the circle.  \texttt{Minimality.lean}: the six minima marked in
Table~\ref{tab:catalogue}, with negative controls that are re-checked on every build.  The rays
and certificates in these files are written by \texttt{ks\_lean\_export.py} and
\texttt{ks\_lean\_minimality.py}.  \texttt{lean/Audit.lean} prints the axioms each theorem
depends on; a transcript of the build and of that audit is in
\texttt{results/lean\_audit\_v11.txt}.
\item \texttt{ks\_modulus2\_circle.py}: Theorem~\ref{thm:circle}, Table~\ref{tab:words},
Appendix~\ref{app:triads}.
\item \texttt{ks\_extras\_v10.py}: Table~\ref{tab:profiles}, the isomorphism tests and the
enumeration of all minimum sets (\texttt{extras\_v10.jsonl}).
\item \texttt{ks\_rigidity\_v10.py}: the flex counts of Section~\ref{sec:rigidity}.
\item \texttt{ks\_half\_circle\_family.py}: the 33-ray family on the circle
$|x|^2=\tfrac12$.  The mode \texttt{exact} does the symbolic computations, including the
identification with Gould and Aravind's Table~3 (\texttt{half\_circle\_family\_exact.json}); the
default mode does the numerical flex counts (\texttt{half\_circle\_family.json}).
\item \texttt{make\_alphabet\_cube\_figures.py}: Figs.~\ref{fig:triptych} and~\ref{fig:generic}.
\end{itemize}

Software: Python~3.13, PySAT (Glucose4, RC2), NetworkX (VF2), NumPy, SciPy (HiGHS), CVXPY
(Clarabel).  Scripts from earlier versions remain in the repository and should not be used to
reproduce this one.

\section{How this version was produced}\label{app:produced}

Computations, proofs and text were produced with a large language model (Claude, Anthropic)
under the author's direction.  Version~9 was then reviewed in full by that model and, separately,
by a second model (GPT, OpenAI) instructed to act as a hostile examiner with access to the
repository; the two reviews agreed on the main defects and each found errors the other missed.
Version~10 was written after those reviews, and its new theorems were each reviewed again in the
same way.  Agreement between language models is not a
proof.  The safeguards we rely on are the ones listed in the Conclusion: exact arithmetic, a
certifier with a short soundness argument, kernel-checked proofs where we have them, positive
controls (the pool counts reproduce those of earlier versions where
those were right, and two independent constructions of the two-symbol pools agree), and an explicit list of what was not checked.  Theorem~\ref{thm:6n} has a
proof that can be read in ten minutes, and independent checking of it and of
Theorems~\ref{thm:twosymbol} and~\ref{thm:circle} is invited.

\section{The 22 triads of the 45 words}\label{app:triads}
Numbers refer to Table~\ref{tab:words}.  These are the triads that hold at every point of the
circle $|x|^2=2$; the 98 orthogonal pairs are printed by \texttt{ks\_modulus2\_circle.py}.

\smallskip\noindent\raggedright
$\{1,2,9\}$, $\{1,16,21\}$, $\{1,26,37\}$, $\{1,29,34\}$, $\{2,10,11\}$, $\{2,12,41\}$, $\{2,13,40\}$, $\{2,14,33\}$, $\{2,15,32\}$, $\{3,4,9\}$, $\{3,36,38\}$, $\{3,43,44\}$, $\{4,35,39\}$, $\{4,42,45\}$, $\{5,8,9\}$, $\{6,7,9\}$, $\{10,28,31\}$, $\{11,27,30\}$, $\{16,22,23\}$, $\{16,24,25\}$, $\{17,18,21\}$, $\{19,20,21\}$

\end{document}